\documentclass{article}
\usepackage[margin=2.5cm]{geometry}
\usepackage{verbatim}
\usepackage[T1]{fontenc}
\usepackage{amsthm}
\usepackage{amsmath,algorithm,tabularx}
\usepackage{amssymb}
\usepackage{xspace}
\usepackage[noend]{algpseudocode}
\usepackage{comment}
\usepackage{tikz}
\usetikzlibrary{fit}
\usepackage{graphicx}
\usepackage{subcaption}

\usepackage{xcolor}
\usepackage{float}

\newcommand{\ktzero}{KT-$0$\xspace}
\newcommand{\ktone}{KT-$1$\xspace}
\newcommand{\kttwo}{KT-$2$\xspace}
\newcommand{\ktrho}{KT-$\rho$\xspace}
\newcommand{\cen}{\text{\sffamily center}}
\newcommand{\id}{\text{\sffamily ID}}

\newcommand{\n}{\text{\sffamily Nbrs}}
\newcommand{\m}{\text{\sffamily Msg}}
\newcommand{\fa}{\textsc{FindAny}}
\newcommand{\gc}{\textsc{GrowCluster}}
\newcommand{\bdtbt}{\textsc{BuildD2BFSTree}}
\newcommand{\rk}{\textsc{Rank}}
\newcommand{\ot}{\tilde{O}}
\newcommand{\dn}{\textit{danner}}
\newcommand{\bc}{\textsc{BroadCast}\xspace}
\newcommand{\congest}{\textsc{Congest}\xspace}
\newcommand{\local}{\textsc{Local}\xspace}
\newtheorem{theorem}{Theorem}
\newtheorem{corollary}{Corollary}[theorem]
\newtheorem{lemma}[theorem]{Lemma}

\newcommand{\e}{\mathbb{E}}
\newcommand{\calc}{\mathcal{C}}

\makeatletter
\newcommand{\multiline}[1]{%
  \begin{tabularx}{\dimexpr\linewidth-\ALG@thistlm}[t]{@{}X@{}}
    #1
  \end{tabularx}
}
\makeatother

\author{
  Hongyan Ji \\
  University of Iowa\\
  \texttt{hongyan-ji@uiowa.edu}
  \and
  Sriram V. Pemmaraju \\
  University of Iowa \\
  \texttt{sriram-pemmaraju@uiowa.edu}
}
\date{}
\begin{document}

\title{Towards singular optimality in the presence of local initial knowledge}

\maketitle              
\thispagestyle{empty}

\begin{abstract}
 The \textit{Knowledge Till $\rho$} (in short, \ktrho) \congest{} model is a variant of the classical \congest model of distributed computing in which each vertex $v$ has initial knowledge of the radius-$\rho$ ball centered at $v$. 
 The most commonly studied variants of the \congest model are \ktzero \congest in which nodes initially know nothing about their neighbors and \ktone \congest in which nodes initially know the \id s of all their neighbors.
 It has been shown that having access to neighbors \id s (as in the \ktone \congest model) can substantially reduce the message complexity of algorithms for fundamental problems such as \bc and MST.  
 For example, King, Kutten, and Thorup (PODC 2015) show how to construct an MST using just $\ot(n)$ messages in the \ktone \congest model for an $n$-node graph, whereas there is an $\Omega(m)$ message lower bound for MST in the \ktzero \congest model for $m$-edge graphs.
 Building on this result, Gmyr and Pandurangen (DISC 2018) present a family of distributed randomized algorithms for various global problems that exhibit a trade-off between message and round complexity. These algorithms are based on constructing a sparse, spanning subgraph called a \textit{danner}. Specifically, given a graph $G$ and any $\delta \in [0,1]$, their algorithm constructs (with high probability) a danner that has diameter $\ot(D + n^{1-\delta})$ and $\ot(\min\{m,n^{1+\delta}\})$ edges in $\ot(n^{1-\delta})$ rounds while using $\ot(\min\{m,n^{1+\delta}\})$ messages, where $n$, $m$, and $D$ are the number of nodes, edges, and the diameter of $G$, respectively.
 In the main result of this paper, we show that if we assume the \kttwo \congest model, it is possible to substantially improve the time-message trade-off in constructing a danner. Specifically, we show in the \kttwo \congest model, how to construct a danner that has diameter $\ot(D + n^{1-2\delta})$ and $\ot(\min\{m,n^{1+\delta}\})$ edges in $\ot(n^{1-2\delta})$ rounds while using $\ot(\min\{m,n^{1+\delta}\})$ messages for any $\delta\in[0,\frac{1}{2}]$.
 This result has immediate consequences for \bc, spanning tree construction, MST, Leader Election, and even local problems such as $(\Delta+1)$-coloring in the \kttwo \congest model. For example, we obtain a \kttwo \congest algorithm for MST that runs in $\ot(D + n^{1/2})$ rounds, while using only $\ot(\min\{m, n^{1 + 1/4}\})$ messages.
\end{abstract}

\newpage
\setcounter{page}{1} 
\section{Introduction}
The \congest model is a standard synchronous, message-passing model of distributed computation in which each node can send an $O(\log n)$-bit message along each incident edge, in each round \cite{peleg00}.
Algorithms in the \congest model are typically measured by their \textit{round complexity} and \textit{message complexity}. The round complexity of an algorithm is the number of rounds it requires to complete, while the message complexity is the total messages exchanged among all the nodes throughout the algorithm's execution. 
Usually, researchers have focused on studying either the round complexity or the message complexity exclusively.
But more recently, researchers have designed \textit{singularly optimal} algorithms in the \congest model, which are algorithms that \textit{simultaneously} achieve the best possible \textit{round} and \textit{message} complexity.  
An excellent example of a singularly optimal algorithm is Elkin's minimum spanning tree (MST) algorithm~\cite{elkin2020simple} that runs in $O((D + \sqrt{n})\log n)$ rounds, using $O(m\log n + n\log n \cdot \log^* n)$ messages. 
Here $n$, $m$, and $D$ are the number of vertices, the number of edges, and the diameter of the underlying graph.
Since MST has a $\tilde{\Omega}(D + \sqrt{n})$ round\footnote{We use $\ot(\cdot)$ to absorb $\text{polylog}(n)$ factors in $O(\cdot)$ and $\tilde{\Omega}(\cdot)$ to absorb $1/\text{polylog}(n)$ factors in $\Omega(\cdot)$.} complexity lower bound \cite{PelegRubinovichSICOMP2000,sarma2012distributed} and an $\Omega(m)$ message complexity lower bound \cite{awerbuch1990trade,KuttenPPRTJACM2015} in the \congest model, Elkin's algorithm is singularly optimal, up to logarithmic factors.

The story becomes more nuanced if we consider the \textit{initial knowledge} that nodes have access to.
The upper and lower bounds cited above are in the \textit{{\bf K}nowledge {\bf T}ill radius 0} (in short, \ktzero) variant of the \congest{} model (aka \textit{clean network} model), in which nodes only have initial knowledge about themselves and have no other knowledge, even about neighbors. 
In the \textit{{\bf K}nowledge {\bf T}ill radius 1} (in short, \ktone) variant of the \congest{} model, nodes initially possess the \id s of all their neighbors.
Round complexity is not sensitive to the distinction between \ktzero \congest and \ktone \congest because nodes can spend 1 round to share their \id s with all neighbors.
However, message complexity is known to be quite sensitive to this distinction. Specifically, the $\Omega(m)$ message complexity lower bound for MST mentioned above only holds in the \ktzero \congest model. In fact, in the \ktone \congest model, King, Kutten, and Thorup (henceforth, KKT)~\cite{king2015construction} presented an elegant algorithm capable of constructing an MST using just $\ot(n)$ messages, while running in $\ot(n)$ rounds.
The $\tilde{\Omega}(D + \sqrt{n})$ round complexity lower bound \cite{sarma2012distributed} for MST mentioned above holds even in \ktone \congest model.
So for an MST algorithm to be singularly optimal in the \ktone \congest model, it would need to run in $\ot(D + \sqrt{n})$ rounds while using $\ot(n)$ messages ($\Omega(n)$ is a trivial message lower bound for MST). The KKT algorithm matches the lower bound of message complexity but is far from reaching the lower bound of round complexity. In fact, a fundamental open question in distributed algorithms is whether it is possible to design  a singularly optimal algorithm for MST in the \ktone \congest model.

Important progress towards \textit{possible} singular optimality for MST in the \ktone \congest model was made by Ghaffari and Kuhn \cite{ghaffari2018distributed} who presented that MST could be solved in $\ot(D + \sqrt{n})$ rounds, and uses  $\ot(\min\{m, n^{3/2}\})$ messages. 
Gymr and Pandurangan \cite{GmyrP18} who also showed the same results at the same time, where MST could be solved in $\ot(D + n^{1-\delta})$ rounds using $\ot(\min\{m, n^{1+\delta}\})$ messages w.h.p.\footnote{We use ``w.h.p.'' as short for ``with high probability'', representing probability at least $1 - 1/n^c$ for constant $c \ge 1$.}
for any $\delta \in [0, \frac{1}{2}]$. Note that by setting $\delta = 1/2$ in this result, one can obtain an algorithm that is round-optimal, i.e., takes $\ot(D + \sqrt{n})$ rounds, and uses  $\ot(\min\{m , n^{3/2}\})$ messages. 
And by setting $\delta = 0$, we can recover the KKT result.
It is worth emphasizing that, in general and more specifically for MST, singular optimality may not be achievable and instead, we have to settle for a trade-off between messages and rounds. In fact, it is possible that the message-round trade-off shown by Gmyr and Pandurangan is optimal, though showing this is an important open question.
The Gmyr-Pandurangan result for MST is based on a randomized algorithm that computes, for any $\delta\in[0,1]$, a sparse, spanning subgraph that they call a \dn\ that has diameter $\ot(D+n^{1-\delta})$ and $\ot(\min\{m,n^{1+\delta}\})$ edges. 
This algorithm runs in $\ot(n^{1-\delta})$ rounds while using $\ot(\min\{m,n^{1+\delta}\})$ messages.
Once the danner is constructed, it essentially serves as the sparse ``backbone'' for efficient communication. As a result, Gymr and Pandurangan obtain round-message tradeoffs not just for MST but for a variety of other global problems such as leader election (LE), spanning tree construction (ST), and \bc. They solve all of these problems in the \ktone \congest model in $\ot(D + n^{1-\delta})$ rounds, using $\ot(\min\{m,n^{1+\delta}\})$ messages for any $\delta \in [0, 1]$. 

An orthogonal direction was previously studied in the seminal paper of Awerbuch, Goldreich, Peleg, and Vainish \cite{awerbuch1990trade} (henceforth, AGPV).
They generalized the notion of initial knowledge and established tradeoffs between the volume of initial knowledge and the message complexity of algorithms.
For any integer $\rho \ge 0$, in the \textit{Knowledge Till $\rho$} (in short, \ktrho) \congest{} model, each node $v$ is provided initial knowledge of (i) the $\id$s of all nodes at distance at most $\rho$ from $v$ and (ii) the neighborhood of every vertex at distance at most $\rho$-1 from $v$. 
The \ktzero and \ktone variants of the \congest model can be viewed as the most commonly considered special cases of the \ktrho \congest model, with $\rho = 0$ and $\rho = 1$ respectively.
AGPV showed a precise tradeoff between $\rho$, the radius of initial knowledge, and the message complexity of global problems such as \bc. 
Specifically, they showed that in the \ktrho \congest model, \bc can be solved using $O(\min\{m, n^{1+c/\rho}\})$ messages, for some constant $c$. 
The main drawback of the AGPV message upper bound is that the \bc algorithm that achieves the $O(\min\{m, n^{1+c/\rho}\})$ message-upper-bound requires $\Omega(n)$ rounds in the worst case. 
Hence, the algorithm exhibits a round complexity significantly exceeding the optimal $O(D)$ rounds required for \bc.

\noindent
\textbf{Main Results.} 
As illustrated by the above discussion, our understanding of singular optimality for global problems in the presence of initial knowledge is severely limited.
Motivated specifically by the results of AGPV \cite{awerbuch1990trade} and those of Gmyr and Pandurangan \cite{GmyrP18}, we consider the design of distributed algorithms for global problems in \kttwo \congest model.
Our main contribution is showing that the round-message tradeoff shown by Gmyr and Pandurangan in the \ktone \congest model can be substantially improved in the \kttwo \congest model. Specifically, we show the following result.

\medskip

\noindent\fbox{%
    \parbox{0.95\columnwidth}{%
\textbf{Main Theorem.} 
There is a danner algorithm in the \kttwo \congest model that runs in $\ot(n^{1-2\delta})$ rounds, using $\ot(\min\{m, n^{1+\delta}\})$ messages w.h.p.
The \dn\ constructed by this algorithm has diameter $\ot(D+n^{1-2\delta})$ and $\ot(\min\{m, n^{1+\delta}\})$ edges w.h.p.
}
}

\medskip

\noindent
Like Gmyr and Pandurangan, we obtain implications of this danner construction for various global problems.
\begin{itemize}
\item We show that \bc, LE, and ST can be solved in $\ot(D+n^{1-2\delta})$ rounds, while using $\ot(\min\{m,n^{1+\delta}\})$ messages 
for any $\delta\in[0,\frac{1}{2}]$.
\item MST can be solved in $\ot(D+n^{1-2\delta})$ rounds, while using $\ot(\min\{m,n^{1+\delta}\})$ messages w.h.p., for $\delta \in [0, \frac{1}{4}]$.
\end{itemize}
Somewhat surprisingly, using recent results of \cite{PaiPPRPODC2021}, we show that even a local problem such as \textit{$(\Delta+1)$-coloring} can benefit from a more efficient danner construction.
In \cite{PaiPPRPODC2021}, the authors present a $(\Delta+1)$-coloring algorithm in the \ktone \congest model that 
uses $\ot(\min\{m,n^{1.5}\})$ messages, while running in $\ot(D + \sqrt{n})$ round.
Using our danner construction algorithm in the \kttwo \congest model, we show how to solve $(\Delta+1)$-coloring in $\ot(D+n^{1-2\delta})$ rounds, while using $\tilde{O}(\min\{m,n^{1+\delta}\})$ messages for any $\delta\in[0,\frac{1}{2}]$.

\smallskip

\noindent
Some specific instantiations of our result are worth considering.
\begin{description}
    \item[$\delta =  1/4$:] \bc, LE, ST, and MST can all be solved in $\ot(D+ \sqrt{n})$ rounds while using  $\ot(\min\{m,n^{1+\frac{1}{4}}\})$ messages. For MST this is round-optimal, even in the \kttwo \congest model\footnote{This follows by observing that the communication-complexity-based lower bound argument \cite{sarma2012distributed} works in the \ktrho 
    \congest{} model for any $\rho \le 2\log \frac{\sqrt{2n}}{4} + 2$, where $n$ is the size of the network, and thus the $\tilde{\Omega}(D + \sqrt{n})$ round lower bound for MST holds even in the \ktrho \congest{} model, for $\rho \le 2\log \frac{\sqrt{2n}}{4} + 2$.} 
    while using significantly fewer messages than the best-known result in the \ktone \congest model \cite{ghaffari2018distributed,GmyrP18}.
    \item[$\delta =  1/3$:] \bc, LE, and ST can all be solved in $\ot(D+ n^{1/3})$ rounds while using  $\tilde{O}(\min\{m,n^{1+\frac{1}{3}}\})$ messages. 
    In the \ktone \congest model, if we use the Gmyr-Pandurangan result \cite{GmyrP18} to match the rounds in this result, we end up using 
    $\tilde{O}(\min\{m,n^{1+\frac{2}{3}}\})$ messages, and if we match the messages in this result
    we end up using $\tilde{O}(D + n^{2/3})$ rounds.
    \item[$\delta =  1/2$:]  \bc, LE, and ST can all be solved in near-optimal $\ot(D)$ rounds while using $\ot(\min\{m, n^{1+1/2}\})$ messages.
\end{description}

It is also important to place our result in the context of implications we can obtain using the results of Derbel, Gavoille, Peleg, and Viennot \cite{derbel2008locality}.
This paper presents a deterministic distributed algorithm that, given an integer $k \ge 1$, constructs in $k$ rounds a $(2k-1)$-spanner with $O(k \cdot n^{1+1/k})$ edges for every $n$-node unweighted graph. This algorithm works in the \local model, which is very similar to the \congest model, except that messages in the \local model can be arbitrarily large in size.
Now note that a $k$-round algorithm in the \local model can be executed using 0 rounds and 0 messages (i.e., completely through local computation) if nodes are provided radius-$k$ knowledge initially.
This implies that in the \kttwo \congest model, a 3-spanner with $O(n^{1+1/2})$ edges can be constructed without communication. One can then use this 3-spanner as a starting point for the various global tasks mentioned above and obtain results that roughly match what we obtain by setting $\delta = 1/2$. 
Specifically, for \bc, LE, and ST, this approach also yields $\ot(D)$ rounds while using $\ot(\min\{m, n^{1+1/2}\})$ messages.
However, this approach does not yield any of our results that use fewer messages, which we obtain by using values of $\delta < 1/2$.
Furthermore, this approach does not improve the message and round complexity results for MST, already known in the \ktone \congest{} model.

\subsection{\ktrho \congest Model}
\label{subsection:model}

We work in the fault-free, message-passing, synchronous distributed computing model, known as the \congest{} model \cite{peleg00}. In this model, the input graph $G = (V, E)$, $n=|V|$, $m=|E|$, also serves as the communication network. Nodes in the graph are processors, and each node has a unique \texttt{ID} drawn from a space whose size is polynomial in $n$. Edges serve as communication links. Each node can send a possibly distinct $O(\log n)$-bit message per edge per round.  
We further classify the \congest{} model based on the amount of initial knowledge nodes have. For any integer $\rho \ge 0$,
we define the \textit{Knowledge Till $\rho$} (in short, \ktrho) \congest{} model as the \congest{} model in which each node $v$ is provided initial knowledge of (i) the IDs of all nodes at distance at most $\rho$ from $v$ and (ii) the neighborhood of every vertex at a distance at most $\rho$ - 1 from $v$. 
Thus, in the \ktzero \congest{} model, nodes do not know the \texttt{ID}s of neighbors. It is assumed that if a node $v$ has degree $d$, then the $d$ incident
edges are connected to $v$ via ``ports'' numbered arbitrarily from 1 through $d$.
In the \ktone \congest{} model, nodes initially know the \texttt{ID}s of neighbors but don't know anything more about their neighbors.
In the rest of the paper, we assume that $\rho \le D$, where $D$ is the diameter of $G$. If $\rho > D$, then every vertex knows $G$ completely at the start, and all problems become trivial in the \ktrho \congest{} model.

\subsection{Challenges, Approach, and Techniques}
\label{subsection:approach}
Our approach combines ideas from the well-known spanner algorithm of Baswana and Sen~\cite{baswana2007simple} with some ideas proposed by Gmyr and Pandurangan~\cite{GmyrP18}, which in turn depend on novel techniques proposed by KKT \cite{king2015construction}. 
In the sequential (or centralized) setting, given an edge-weighted graph $G=(V,E)$ and any integer $k \ge 1$, the Baswana-Sen algorithm computes a $(2k-1)$-spanner with $O(k \cdot n^{1+\frac{1}{k}})$ edges in expected $O(k \cdot m)$ time, where $m$ is the number of edges.
The algorithm consists of two \textit{phases}. In Phase 1, over a course of $k-1$ \textit{iterations}, clusters are subsampled with probability $n^{-1/k}$ and then grown. This process establishes disjoint clusters, each resembling a rooted tree with a center. Initially, each vertex is by itself an individual cluster.
In Phase 2, clusters are merged; this involves each vertex selecting a minimum-weight edge to each adjacent cluster and incorporating it into the spanner.
The natural distributed implementation of the Baswana-Sen algorithm requires $O(k^2)$ rounds and uses $O(k \cdot m)$ messages in the \ktzero \congest model.
This is clearly too message-inefficient for our purposes.
The bottleneck in the  Baswana-Sen algorithm is that for each sampled cluster to grow, it needs to inform all its neighbors that it has been sampled.
More specifically this challenge appears in two forms.

\begin{description}
\item[Too many clusters:] In the early iterations (in Phase 1) of the Baswana-Sen algorithm, there are too many clusters. We will end up using too many messages if every cluster tries to inform every neighbor. This issue appears in the first iteration, in which each cluster is an individual node. For example, suppose that we want to produce a spanner with $O(n^{1 + \frac{1}{3}})$ edges. Producing such a spanner would imply that downstream tasks such as \bc can be completed using $O(n^{1 + \frac{1}{3}})$ messages. So we pick $k = 3$ and the Baswana-Sen algorithm samples clusters with probability $n^{-1/3}$ in the first iteration. This yields $\Theta(n^{2/3})$ clusters w.h.p.~and if each cluster sent messages to all neighbors, we could end up using $\Omega(n^{1+\frac{2}{3}})$ messages, well above our target of $O(n^{1 + \frac{1}{3}})$ messages.

\item[Redundant messages:] Even if we were able to circumvent the above issue, there is a second and even more challenging obstacle. Suppose we have reached a point where the clusters have grown to trees of some constant diameter and the number of clusters is small enough that each cluster is permitted to send $n$ messages informing neighbors.  
In this setting, if each node in a cluster sent messages to neighbors outside the cluster in an uncoordinated manner, we could end up sending up to $\Omega(n^2)$ messages because each neighbor of the cluster could receive the same message from multiple nodes in the cluster. Removing these redundant messages requires coordination within the cluster before sending messages, but the coordination itself can be quite message-costly.
\end{description}
We overcome these challenges in a variety of ways.
First, we design a randomized estimation procedure that clusters can use to estimate if a neighbor $w$ will hear from \textit{other} sampled clusters. 
There is of course no need to send $w$ a message if it is estimated that someone else will communicate with $w$.
This estimation procedure critically depends on 2-hop initial knowledge. It allows clusters to communicate selectively with neighbors while still guaranteeing that every node $w$ that is a neighbor of a sampled cluster joins one such cluster.
To circumvent the challenge of redundant messages, we introduce two new subroutines, 
for growing a ``star'' cluster $C$ that is both round and message efficient.
Both subroutines critically depend on using 2-hop initial knowledge for their (simultaneous) round and message efficiency.
For example, the $\gc(C)$ subroutine (see Section \ref{subsection:fastSubroutines}) takes as input a ``star'' cluster $C$ with $N$ neighbors and grows the ``star'' by adding one edge from $C$ to each neighbor. Our implementation requires $O(\sqrt{N})$ rounds while using a total of $O(N)$ messages, which is linear in the size of the constructed cluster. 
The estimation procedure and subroutines for cluster growing may be of independent interest to anyone designing efficient algorithms in the \kttwo \congest model.

Another technique we use is to allow surplus messages in early iterations, which even though not necessary in the early iterations, can improve message complexity in later iterations when combined with estimation procedures.

\subsection{Related Work}
\label{subsection:relatedWork}
While the current paper focuses only on synchronous models, we note that there is a growing body of related work in asynchronous models of distributed computation. 
In \cite{KuttenPPRTJACM2015}, a singularly near-optimal randomized leader election algorithm for general synchronous networks in the \ktzero \congest model is presented. 
This result was extended to the asynchronous \ktzero \congest model in \cite{KuttenMPPDISC2020,KuttenMPPDISC2021}.
Even for MST, there has been recent work on singularly optimal randomized MST algorithms in the asynchronous \ktzero \congest{} model \cite{DufoulonKMPPDISC2022}.
This paper also contains an asynchronous MST algorithm that is sublinear in both time and messages in the  \ktone \congest{} model. 

Since a danner is a relaxation of a spanner, it is worth mentioning a recent lower bound result for spanner construction due to Robinson \cite{RobinsonSODA2021}.
He considers the \ktone \congest model and shows that any algorithm running in $O(\text{poly}(n))$-time must send at least $\tilde{\Omega}(\frac{1}{t^2} n^{1+1/2t})$ bits to construct a $2t-1$-spanner. It would be interesting to determine if this type of spanner lower bound can be extended to danner construction.

Earlier, we mentioned the work of Derbel, Gavoille, Peleg, and Viennot \cite{derbel2008locality}. Another immediate implication of this work is that, for an integer $\rho \ge 1$, it is possible to construct a $(2\rho-1)$-spanner with $O(\rho \cdot n^{1+1/\rho})$ edges
using no communication in the \ktrho \congest model. 
One can then use this $(2\rho-1)$-spanner as a starting point for various global tasks mentioned earlier (\bc, LE, ST, MST).
For example, this implies that for \bc, LE, and ST there are algorithms in the \ktrho \congest model that run in $O(\rho \cdot D)$ rounds, using $O(\rho \cdot n^{1+1/\rho})$ messages. For MST, the corresponding algorithm in the \ktrho \congest model would run in $\ot(\rho (\cdot D + \sqrt{n}))$ rounds, using $O(\rho \cdot n^{1+1/\rho})$ messages. Setting $\rho = \Theta(\log n)$, this would yield an MST algorithm in the \congest model in which nodes have radius-$\Theta(\log n)$ initial knowledge, running in near-optimal $\ot(D + \sqrt{n})$ rounds, using near-optimal  $\ot(n)$ messages.

As mentioned earlier, AGPV showed an upper bound of $O(\min\{m, n^{1+c/\rho}\})$ on the message complexity of \bc on an $n$-vertex, $m$-edge graph in the \ktrho \congest model. But, this algorithm can take $\Omega(n)$ rounds in the worst case because the AGPV algorithm starts by performing a deterministic sparsification step that takes 0 rounds (i.e., only local computation is needed by this algorithm), reduces the number of edges in the graph to $O(\min\{m, n^{1+c/\rho}\})$, but can produce graphs with $\Omega(n)$ diameter.
We present one such worst-case example in detail in the appendix.

\subsection{Notation and Definitions}
Let $\n(w)$ denote the set of neighbors of node $w$ and
let $\n_2(w)$ denote the set of 2-hop neighbors of node $w$.
A \textit{cluster} $C=(V(C),E(C))$ is a connected subgraph of graph $G=(V,E)$. 
All clusters considered in this paper will be constant-diameter trees.
Furthermore, every cluster constructed by algorithms in this paper will start as a single node and then grow over the course of the algorithm.
For a cluster $C$, we will use $\cen(C)$ to denote the (unique) oldest node in a cluster and we use
the \id\ of $\cen(C)$ as the \id\ of cluster $C$; we will use the notation $\id_C$ to denote the \id\ of $C$. 
Let $\n(C)$ denote the set of neighboring vertices of cluster $C$, i.e., $\n(C) \cap C = \emptyset$ and every $w\in \n(C)$ has a neighbor in $C$.

\subsection{Organization}

We start the rest of the paper by describing efficient subroutines in the \kttwo \congest model for 3 key tasks (Section \ref{subsection:fastSubroutines}). These subroutines are both round and message efficient and use 2-hop initial knowledge crucially for their efficiency.
We then describe our main danner algorithm and its analysis (Section \ref{section:danner}), followed by various applications of our danner algorithm to problems such as \bc, LE, ST, MST, and $(\Delta+1)$-coloring.

\section{Fast subroutines in the \kttwo \congest model}
\label{subsection:fastSubroutines}

In this section, we identify 3 key tasks that can be implemented in a round- and message-efficient
manner due to access to initial 2-hop knowledge. It is unclear how to execute these tasks efficiently without initial 2-hop knowledge, e.g., in the \ktone \congest model.
For each of the 3 tasks, we present subroutines that are round- and message-efficient.

\begin{description}
\item
[Rank in neighbor's neighborhood:] For a given node $v$ and a given neighbor $w \in \n(v)$, we need to calculate the rank of its identifier ($\id_v$) within the neighborhood of $w$.
We use $\rk(v, w)$ to denote the subroutine that completes this task in the \kttwo \congest model.
It is immediate that $\rk(v, w)$ completes this task in 0 rounds, using 0 messages because $v$ has all the information it needs within its 2-hop initial knowledge. 
One might think that this task has been efficiently completed in the \ktone \congest model as well.
In a sense, this is true because in the  \ktone \congest model node $v$ can simply ask $w$ to compute the rank of $\id_v$ in $w$'s neighborhood; this would take 2 rounds and 2 messages.
Unfortunately, even this is too inefficient for our purposes because the $\rk(v, w)$ subroutine will be used by $v$ as a filter to determine whether $v$ even needs to communicate with $w$. 

\item[Depth-2 BFS tree:] Given a node $v$, our task is to efficiently construct a depth-2 BFS tree rooted at $v$.
We now define a subroutine $\bdtbt(v)$ in the \kttwo \congest model that can complete this task in 2 rounds, using $O(K)$ messages, where $K = |\n_2(v)|$.
In other words, our goal is to use constant rounds and bound the number of messages by the size of the depth-2 BFS tree that is constructed.

\smallskip
\begin{enumerate}
\item Node $v$ sends a message to each neighbor $w \in \n(v)$ and the edges $\{v, w\}$ are added to the output tree.

\item Using 2-hop initial knowledge, each node $w \in \n(v)$ can locally compute the set $\n(v)$.
Then node $w \in \n(v)$ can use 2-hop initial knowledge to select a subset of neighbors to send messages to. 
Specifically, node $w$ sends a message to a neighbor $x$ iff $\id_w$ is the lowest \id\ among the \id s of nodes in $\n(v) \cap \n(x)$. Node $w$ can check whether it satisfies this condition using local computation on its initial 2-hop knowledge.
\end{enumerate}

Note that it is possible to construct the second level of the depth-2 BFS tree rooted at $v$ by using a standard ``flooding'' algorithm in which each node $w \in \n(v)$ sends a message to each of its neighbors.
However, in the worst case, this could take $\Omega(K^2)$ messages. Using 2-hop knowledge allows for a much more message-efficient algorithm, while using constant number of rounds.

\begin{lemma}
For any $v \in V$, the subroutine $\bdtbt(v)$ runs in $O(1)$ rounds, using $O(|\n_2(v)|)$ messages. 
\end{lemma}
\begin{proof}
Step 1 takes 1 round, with $O(|\n(v)|)=O(|\n_2(v)|)$ messages, given that each neighbor $w\in\n(v)$ receives a messages from the node $v$. Step 2 takes 1 round, with $O(|\n_2(v)|-|\n(v)|)=O(|\n_2(v)|)$ messages because each node $w$ that is 2-hop away from node $v$ receives 1 message from a node $w\in\n(v)$.   
\end{proof}

\item[Growing a ``star'' cluster:] Consider a cluster $C$ that is ``star'' graph. In other words, $\cen(C)$ is some vertex $v \in V$, the rest of the vertices satisfy $V(C) \setminus \{\cen(C)\} \subseteq \n(v)$, and there are $|V(C)|-1$ edges, from $\cen(C)$ to each node in $V(C) \setminus \{\cen(C)\}$.  Note that $\n(C) \subseteq \n_2(\cen(C))$ and it is possible for $|\n(C)|$ to be much smaller than $|\n_2(\cen(C))|$.
Let $N = |\n(C)|$.
We need to complete this task efficiently: grow the cluster $C$ by adding an edge from $C$ to each of its $N$ neighbors. 

We now define an efficient subroutine $\gc(C)$ for this task that uses $O(\sqrt{N})$ rounds and $O(N)$ messages. 
As with the previous subroutines, 2-hop initial knowledge plays a critical role in achieving these round and message complexities.
Note that if $N \approx |\n_2(\cen(C))|$, we can simply call the $\bdtbt(\cen(C))$ subroutine defined above to complete this task in
$O(1)$ rounds and $O(N)$ messages. So the challenge is in designing an efficient algorithm (in terms of $N$) even when $N \ll |\n_2(\cen(C))|$.

\medskip
Given a rooted tree $T$ and a node $u$ in $T$, we use $ch_{T}(u)$ to denote the set of children of $u$ in $T$.
We now describe our algorithm for $\gc(C)$ (See Figure~\ref{fig:subroutineGS} for illustration.).   
\begin{enumerate}
\item (\textbf{Local computation.}) $\cen(C)$ uses 2-hop initial knowledge and knowledge of $V(C)$ to locally construct a tree $T$ obtained by adding edges to $C$, where each added edge is from a node $w \in \n(C)$ to a neighbor of $w$ in $C$ with minimum \id.
Viewing $T$ as a tree rooted at itself, $\cen(C)$ classifies each of its children $u$ as a \textit{low-degree} node if $|ch_T(u)| \le \sqrt{N}$; the rest of its children are classified as \textit{high-degree} nodes. \\
\textbf{Notation:} We use $LDC$ to denote the set of low-degree children of $\cen(C)$ and similarly $HDC$ as the set of high-degree children of $\cen(C)$.
\item  To each node $u \in LDC$, $\cen(C)$ sends the \id s of all nodes in $ch_T(u)$, one $\id$ at a time. 
To each node $u \in HDC$, $\cen(C)$ sends \id s of all nodes in $HDC$, again one \id\ at a time.
\item Each node $u \in LDC$ sends message $\m_1$ to each node $v \in ch_T(u)$.
\item Each node $u \in HDC$ sends message $\m_2$ to each node $v \in \n(u)$, if $u$ is the node with the smallest \id\ in $HDC \cap \n(v)$.
\item Each node $v \not\in V(C)$ adds edge $\{u, v\}$ to the output, where $u$ is the node with smallest \id\ from which it has received a message.
\end{enumerate}

\begin{lemma}
\label{lemma:GCComplexity}
Let $T$ be the tree locally constructed by $\cen(C)$ in Step 1 of subroutine $\gc(C)$.
The subroutine $\gc(C)$ runs in $O(\sqrt{N})$ rounds and uses $O(N)$ messages and at the end of the subroutine edges in $T$ that are not already in $C$ are added to the output.    
\end{lemma}
\begin{proof}
We begin by bounding the round complexity of the subroutine. In the first step, local computation is performed using initial 2-hop knowledge. In Step 2 it incurs a round complexity of $O(\sqrt{N})$ because $\cen(C)$ transmits a maximum of $\sqrt{N}$ $\id$s individually to each node $u \in LDC$. Simultaneously, $\cen(C)$ consumes at most $\sqrt{N}$ rounds in sending, individually, up to $\sqrt{N}$ $\id$s to each node $u \in HDC$. This is because $\cen(C)$ has at most $\sqrt{N}$ high-degree children. Steps 3, 4, and 5 each can be executed in a single round. 

We now bound the message complexity in the following manner. Step 1 is local computation with 0 message. Step 2 uses $O(N)$ messages because $\cen(C)$ sends at most $N$ messages to nodes in $LDC$, and at most $N$ messages to nodes in $HDC$. This stems from the same reasoning as before, where there are at most $\sqrt{N}$ nodes in $HDC$. In Steps 3, 4, and 5, the message count is bounded by $O(N)$ messages given that $|\n(C)|=N$.

Consider a node $w \in \n(C)$. If the only neighbors of $w$ in $C$ belong to LDC, then $w$ is guaranteed to receive a message along the incident edge in $T$.
Similarly, if the only neighbors of $w$ in $C$ belong to HDC, then $w$ is guaranteed to receive a message along the incident edge in $T$ because the message to $w$ from the nodes in HDC is sent from the node in HDC with smallest \id.
If $w$ has some neighbors in LDC and some in HDC, it could receive 2 messages. But, because $w$ picks a neighbor with lower \id, again the edge in $T$ is added to the output.
\end{proof}

\begin{figure}[ht]
  \begin{subfigure}{0.5\textwidth}
    \centering
    \includegraphics[width=0.8\textwidth]{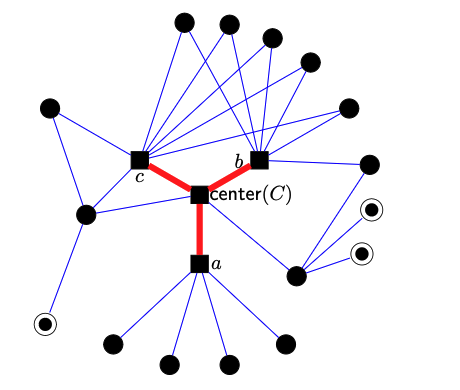}
  \end{subfigure}%
  \begin{subfigure}{0.5\textwidth}
    \captionsetup{justification=raggedright, singlelinecheck=false, font=footnotesize}
    \includegraphics[width=0.8\textwidth]{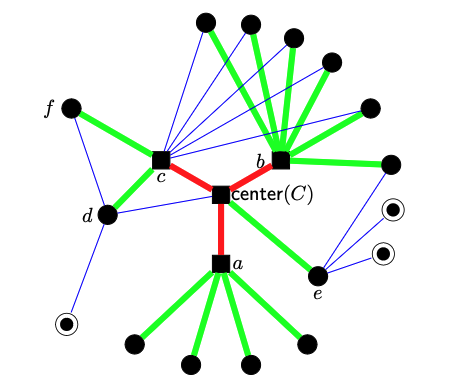}
  \end{subfigure}
  \caption{Cluster $C$ is represented as a star graph with four nodes (depicted as rectangles) and three edges (illustrated as thick, red line segments). The cluster $C$ is connected to $N=13$ neighbors, displayed as black disks. Note that the 3 nodes depicted with concentric circles are not neighboring to $C$. The tree $T$ locally constructed by $\cen(C)$ consists of thick (red) edges shown on the left plus the thick (green) edges shown on the right. In $T$, the connection between node $d$ and node $c$ is established, bypassing $\cen(C)$ due to the fact that $\id_c$ is smaller than the $\id$ of $\cen(C)$. The set of low-degree children, denoted as $LDC$, is $\{c, e\}$, while the set of high-degree children, denoted as $HDC$, is $\{a, b\}$. This classification is based on the cardinalities of the children sets in $T$: $|ch_T(c)|=2$, $|ch_T(e)|=0$, $|ch_T(a)|=4$, $|ch_T(b)|=6$ compared to $\sqrt{N}=\sqrt{13}$. Consequently, $\cen(C)$ transmits the $\id$s of nodes $d$ and $f$ to node $c$, while nodes $a$ and $b$ receive the $\id$s of nodes $a$ and $b$ from $\cen(C)$.}\label{fig:subroutineGS}
\end{figure}

\end{description}

\section{Distributed Danner Construction in the \kttwo \congest model}
\label{section:danner}
As mentioned earlier, our danner algorithm is inspired by the celebrated Baswana-Sen spanner algorithm \cite{baswana2007simple}. 
For any integer $k \ge 1$, this algorithm constructs a $(2k-1)$-spanner
by subsampling and growing clusters for $k-1$ iterations and then merging them.
We consider this algorithm for $k = 3$ and implement the 2 iterations of the Baswana-Sen algorithm in a round- and message-efficient manner by leveraging 2-hop initial knowledge in a fundamental way.
These correspond to the two \textit{cluster growing} phases (described in Algorithm~\ref{alg:part1} and Algorithm~\ref{alg:part2hd}, \ref{alg:part2ld}).
We use two versions of Phase 2 of the cluster growing algorithm, one for high $\delta$ ($\delta \in (1/3, 1/2]$) and one for low $\delta$
($\delta \in [0, 1/3]$).
It is unclear how to implement the merging step of the Baswana-Sen algorithm in a round- and message-efficient manner. 
So instead, we use 
ideas similar to those used by Gmyr and Pandurangan \cite{GmyrP18} for merging the clusters. This \textit{cluster merging} phase is described in Algorithm~\ref{alg:part3}. 
Our main technical contributions are in the cluster growing phases, so we explain these fully with all proofs provided. Since the cluster merging phase uses ideas similar to those used by Gmyr and Pandurangan \cite{GmyrP18}, we sketch this phase and refer the reader to \cite{GmyrP18} for more details.  

\subsection{Cluster Growing: Phase 1}
\label{subsection:CG1}

Algorithm~\ref{alg:part1} starts with a set $\calc_0$ of initial clusters created with each node by itself being a cluster. We then sample each cluster with probability $n^{-\delta}$ and create a set $\calc_1$ of sampled clusters. Let $U$ denote the set of nodes not in sampled clusters, i.e., $U := \{v \in V \mid \{v\} \not\in \calc_1\}$. The rest of the algorithm aims to ``grow'' these sampled clusters by having unsampled nodes join neighboring sampled clusters. Since there are $\Theta(n^{1-\delta})$ sampled clusters w.h.p., it is not message efficient for each cluster $C \in \calc_1$ (which is just a node at this point) to communicate with all neighbors $w\in\n(C)$. 
Instead, cluster $C \in \calc_1$ uses the rank computation subroutine (see Section \ref{subsection:fastSubroutines}) to reduce the message complexity of this step.
Specifically, cluster $C$ first checks (in Step 4) if its \id\ belongs to the smallest $\lceil {n^{2\delta}\rceil}$ \id s of \textit{neighbors} of $w$\footnote{Note that this step requires 2-hop knowledge; all steps in our algorithms which assume 2-hop knowledge are highlighted in gray.}.
In Step 5, the sampled cluster $C \in \calc_1$ sends a message $\m_1(\id_C)$ just to those neighbors $w$ who pass this check. For $w\in V$, let $M_1(w)$ denote a set of tuples, including the sender's ID and the edge which the sender uses to send a message to $w$, i.e., $M_1(w) := \{(\id_C, e) \mid w\text{ receives }\m_1(\id_C)\text{ along edge }e\}$. 
The choice of the $\lceil {n^{2\delta}\rceil}$-sized ``bucket'' of smallest \id\ neighbors of $w$ is critical in ensuring two properties we need: (i) every node $w$ receives $\ot(n^\delta)$ messages (Lemma \ref{lem:alg1NS}) and (ii) every node $w$ that does not receive a message has low, i.e., $\ot(n^\delta)$, degree (Lemma \ref{lem:part1lowdegree}).
Subsequently, every node $w$ that does not belong to a sampled cluster, can take one of two actions.
If $w$ receives a message from a neighboring sampled cluster, it joins the sampled cluster $S_w$ with minimum $\id$ among all sampled clusters from which it receives a message (Steps 9-11). When $w$ joins a cluster, the edge connecting $w$ to the cluster is added to the cluster and the danner $H$. 
For nodes $w$ that do not receive any message from a sampled cluster, all incident edges of $w$ are added to the danner (Step 13). 
The fact that such nodes are guaranteed to have a low degree is critical to ensuring this step is message-efficient.

\begin{algorithm}
\caption{Cluster Growing: Phase 1}\label{alg:part1}
\begin{algorithmic}[1]
\Require $G=(V,E)$, $\calc_0 = \{\{v\} \mid v \in V\}$, $H = (V, \emptyset)$
\Ensure a set $\calc_{1}$ of clusters, partially constructed danner $H$
\State Independently sample each cluster $C \in \calc_0$ with probability $n^{-\delta}$;
\Statex \textbf{Notation:} $\calc_1 \subseteq \calc_0$ denotes the set of sampled clusters; for each cluster $C = \{v\} \in \calc_1$, $v$ is the $\cen$ of $C$, denoted by $\cen(C)$; $U := \{v \in V \mid \{v\} \not\in \calc_1\}$ is the set of nodes not in sampled clusters
\Statex
\For{$C\in \calc_1$} \Comment{Actions by sampled clusters}
    \For{$w\in \n(C)$}
        \If{\colorbox{gray!20}{
        $\rk(\cen(C), w)\le \lceil {n^{2\delta}\rceil}$
        }}
         \State $\cen(C)$ sends a message $\m_1(\id_C)$ to $w$. 
        \EndIf
    \EndFor
\EndFor
\Statex
\For{$w \in V$} 
    \State $M_1(w) := \{(\id_S, e) \mid w\text{ receives }\m_1(\id_S)\text{ along edge }e\}$
    \If{$w \in U$} \Comment{Actions by nodes not in sampled clusters}
    \If{$M_1(w) \not= \emptyset$} \Comment{Actions by nodes that hear from a sampled cluster}
     \State $S_w := \text{cluster $S$ with minimum \id\ in $M_1(w)$}$. 
     \State $w$ joins the cluster $S_w$, the edge $\{w, \cen(S_w)\}$ is added to cluster $S_w$ and to the danner $H$. 
    \Else \Comment{Actions by nodes that don't hear from a sampled cluster}
     \State $w$ adds all incident edges to $H$ by sending messages along incident edges. 
     \Statex \ \qquad\qquad\textbf{Note:} $w$ becomes inactive and does not participate any further in the algorithm.
    \EndIf
    \EndIf
\EndFor
\end{algorithmic}
\end{algorithm}

We now prove bounds on the round and message complexity of Algorithm \ref{alg:part1} and then prove properties of the output produced by the algorithm.
\begin{lemma}
\label{lem:alg1NS}
    In Algorithm~\ref{alg:part1}, for any $w \in U$, $|M_1(w)| = O(n^{\delta})$ w.h.p.
\end{lemma}
\begin{proof}
    For any $w \in U$, the number of neighbors who are candidates for sending $w$ a message (in Step 5) is at most $\lceil {n^{2\delta}\rceil}$. This is because (in Step 4) only neighbors of $w$ whose \id s are among the smallest $\lceil {n^{2\delta}\rceil}$ \id s communicate with $w$.
    Furthermore, among these $\lceil {n^{2\delta}\rceil}$ neighbors of $w$, only those nodes which are sampled (in Step 1) communicate with $w$. Since this sampling occurs independently with probability $n^{-\delta}$, by a simple application of Chernoff bounds, we see that w.h.p.~node $w$ will receive messages from $O(n^{\delta})$ neighbors.
\end{proof}

\begin{lemma}\label{lem:part1lowdegree}
    In Algorithm~\ref{alg:part1}, for any $w \in U$, if $|M_1(w)| = \emptyset$ then w.h.p.~$|\n(w)| = \ot(n^\delta)$.
\end{lemma}
\begin{proof} 
    Let $N = \min\{|\n(w)|, \lceil n^{2\delta}\rceil\}$. 
    Suppose $|\n(w)| \ge c \cdot n^{\delta} \ln n$ for some constant $c$. Then $N \ge c \cdot n^{\delta} \ln n$.
    Since each sampled neighbor of $w$ with \id\ among the smallest $\lceil {n^{2\delta}\rceil}$ \id s sends $w$ a message, the probability that no node sends $w$ a message is at most $(1-n^{-\delta})^N \le n^{-c}$.   
    Thus if $|\n(w)| \ge c \cdot n^{\delta} \ln n$, the probability that $M_1(w) = \emptyset$ is at most $n^{-c}$. The lemma follows.
\end{proof}
\begin{lemma}\label{lem:part1}
    Algorithm~\ref{alg:part1} runs in $O(1)$ rounds and uses $\ot(n^{1+\delta})$ messages. 
\end{lemma}
\begin{proof}
    We begin by bounding the round complexity of the algorithm. Communication occurs only in Steps 5, 11, and 13; each can be executed in a single round. The other steps only involve local computation, with Step 4 using initial 2-hop knowledge.
    
    We now bound the message complexity in the following manner. 
    In Step 5, the message count is bounded by $O(n^{1+\delta})$ w.h.p., given that each node $w \in U$ receives $O(n^\delta)$ messages w.h.p.~(as demonstrated in Lemma~\ref{lem:alg1NS}). 
    In Step 11 every node $w \in U$ sends a message along a single edge, incurring a total of $O(n)$ messages. In Step 13, the message count is $\ot(n^{1+\delta})$ because each node $w$ for which $M_1(w) = \emptyset$ has $\ot(n^\delta)$ neighbors (as per Lemma~\ref{lem:part1lowdegree}). 
\end{proof}

\begin{lemma} 
\label{lem:algo1Correctness}
After Algorithm~\ref{alg:part1} completes (a) $\calc_1$ contains $\Theta(n^{1-\delta})$ clusters w.h.p.~and every cluster is a star graph, (b) $H$ is a spanning subgraph of $G$ containing all cluster edges and all edges incident on nodes not in clusters, and (c) $H$ contains $\ot(n^{1+\delta})$ edges.
\end{lemma}
\begin{proof}
(a) Since clusters are sampled independently with probability $n^{-\delta}$,  $|\calc_1| = \Theta(n^{1-\delta})$ w.h.p. Each cluster $C \in \calc_1$ is a star graph because it starts off as a single node $\cen(C)$ and then some neighbors $w \in \n(\cen(C))$ and the connecting edges $\{w, \cen(C)\}$ join cluster $C$. 
(b) All cluster edges are added to $H$ during Step 11, and all edges incident on nodes outside clusters are included in $H$ during Step 13.
(c) Each $w$ contributes at most one edge to $H$ in Step 11 for a total of $O(n)$ edges.
Additionally, node $w$ that does not receive a message from a neighboring cluster, adds $\ot(n^\delta)$ edges to $H$ (Lemma~\ref{lem:part1lowdegree}, Step 13). This contributes $\ot(n^{1+\delta})$ edges to $H$.
\end{proof}

\subsection{Cluster Growing: Phase 2}
\label{subsection:CG2}

In Algorithm Cluster Growing: Phase 2 (refer to pseudocode in Algorithm~\ref{alg:part2hd} and \ref{alg:part2ld}), the clusters constructed in Algorithm~\ref{alg:part1} are further subsampled and grown. The details of this algorithm are more complicated than Algorithm~\ref{alg:part1}, so we first describe it at a high level.
At the start of the algorithm, each cluster constructed in Algorithm~\ref{alg:part1} is sampled with probability $n^{-\delta}$. This produces a collection $\calc_2$ of $\Theta(n^{1-2\delta})$ clusters. We show that the clusters that are not sampled can be partitioned into two groups: (i) \textit{high-degree} clusters, which are guaranteed to have a sampled cluster in their neighborhood, and (ii) \textit{low-degree} clusters, which (as the name suggests) are guaranteed to have a small neighborhood, i.e., $\ot(n^{2\delta})$ nodes in their neighborhood.
Each high-degree cluster $C$ connects to a sampled cluster $C'$ in its neighborhood, thus leading to the growth of cluster $C'$.
For each low-degree cluster $C$ and each neighbor $w \in \n(C)$, we add an edge from $C$ to $w$ to the danner. See Figure~\ref{fig:phase2}. We can afford to do this because such clusters have low degrees.  
We now explain the algorithm in more detail. In fact, we have two separate algorithms, one for high $\delta$, i.e., $\delta \in (\frac{1}{3},\frac{1}{2}]$ (Algorithm~\ref{alg:part2hd}), and one for low $\delta$,
i.e., $\delta \in [0, \frac{1}{3}]$ (Algorithm~\ref{alg:part2ld}).
The high $\delta$ algorithm is easier and we explain it first.

\medskip

\noindent
\textbf{High $\delta$ case:}
In this case, each sampled cluster $C \in \calc_2$ can (at least in theory) communicate the fact that it has been sampled to all its neighboring nodes. This is because there are $\Theta(n^{1-2\delta})$ sampled clusters in $\calc_2$ w.h.p.~and for $\delta > 1/3$, $\Theta(n^{1-2\delta}) \times n = \Theta(n^{2-2\delta}) = O(n^{1+\delta})$.
The actual communication is implemented by $\cen(C)$ using the depth-2 BFS tree subroutine described in Section \ref{subsection:fastSubroutines} to build a depth-2 BFS tree and broadcast via this tree to its 2-hop neighborhood (see Step 3).
As established in the description of the depth-2 BFS tree subroutine, all of this takes $O(1)$ rounds and $O(n)$ messages.
If a cluster $C$ is not sampled, and some node $w \in C$ receives a message from a sampled cluster, then $C$ is identified as a high-degree cluster (Step 8). Every high-degree cluster identified in this manner connects to the sampled cluster $C'$ with the lowest \id\ that it hears from (Step 10). As a result, the (non-sampled) cluster $C$ joins sampled cluster $C'$ and an edge via which $C$ heard about $C'$ is added to cluster $C'$ as well as the danner $H$.
If a cluster $C$ is not sampled and it does not hear from a sampled cluster, it is identified as a \textit{low-degree cluster}.
As the name suggests, we show in Lemma \ref{lem:algo2hdldclusters} that w.h.p.~the center of every such low-degree cluster $C$ has only 
$O(n^{2\delta})$ nodes in its 2-hop neighborhood.
Since the total number of clusters is $O(n^{1-\delta})$ w.h.p., each low-degree cluster can afford to communicate with all nodes in its 2-hop neighborhood and add one edge connecting $w$ to $C$, for each $w \in \n(C)$, to the danner $H$ (Step 12).
Again, the actual implementation of this step uses the depth-2 BFS tree subroutine.

\setcounter{algorithm}{1}
\renewcommand{\thealgorithm}{\arabic{algorithm}(a)}
\begin{algorithm}
\caption{Cluster Growing: Phase 2 [High Delta]}\label{alg:part2hd}
\begin{algorithmic}[1]
\Require $\delta\in (\frac{1}{3},\frac{1}{2}]$; $G=(V,E)$, $\calc_1$ and $H$ are the set of clusters and partial danner output by Algorithm \ref{alg:part1}.
\Ensure a set $\calc_{2}$ of clusters, partially constructed danner $H$
\State Independently sample each cluster $C \in \calc_1$ with probability $n^{-\delta}$.
\Statex \textbf{Notation:} $\calc_2 \subseteq \calc_1$ denotes the set of sampled clusters; for each cluster $C \in \calc_2$, each node $w \in V(C) \setminus \{\cen(C)\}$ is a child of $\cen(C)$.
\Statex
    \For{$C\in \calc_{2}$} \Comment{$C$ is sampled.}
        \State \colorbox{gray!20}{$\cen(C)$ uses $\bdtbt(\cen(C))$ to broadcast $\m_2(\id_C)$ to $w\in\n_2(\cen(C))$
        }
    \EndFor
    \Statex
    \For{$C\in \calc_{1} \setminus \calc_2$} \Comment{$C$ is not sampled.}
        \State Each node $w \in V(C)$ computes $M_2(w) := \{(\id_S, e) \mid w\text{ receives }\m_2(\id_S)\text{ along edge }e\}$
        \State Each child $w$ in $C$ with $M_2(w) \not= \emptyset$ sends $\m_3(\min M_2(w))$ to $\cen(C)$
        \State $\cen(C)$ computes $M_3 := \{(\id_S, e) \mid \cen(C)\text{ receives }\m_3(\id_S, e)\}$
        \If{$M_2(\cen(C)) \cup M_3 \not= \emptyset$} \Comment{$C$ is a \textbf{high-degree cluster}}
            \State $\cen(C)$ computes $(\id_{C'}, e) = \min (M_2(\cen(C)) \cup M_3)$
            \State $\cen(C)$ connects to cluster $C'$ via edge $e$; edge $e$ is added to cluster $C'$ and to $H$.    
        \Else \Comment{$C$ is a \textbf{low-degree cluster}}
            \State \colorbox{gray!20}{ $\cen(C)$ uses $\bdtbt(\cen(C))$ and adds edges of the tree to $H$
            } 
        \EndIf
    \EndFor
\end{algorithmic}
\end{algorithm}

\begin{figure}[ht]
  \begin{subfigure}{0.5\textwidth}
    \centering
    \includegraphics[width=0.8\textwidth]{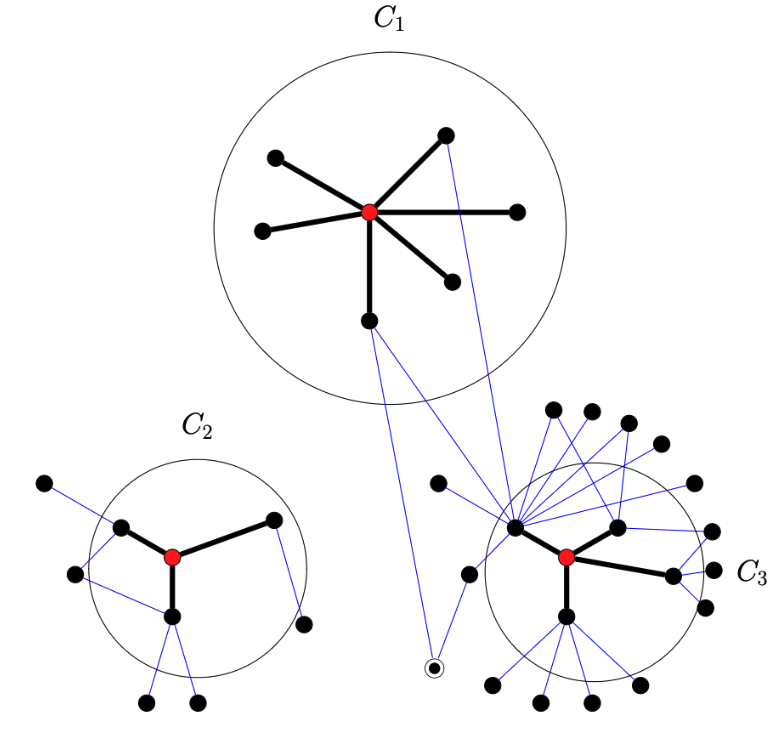}
  \end{subfigure}
  \begin{subfigure}{0.5\textwidth}
    \captionsetup{justification=raggedright, singlelinecheck=false, font=footnotesize}
    \includegraphics[width=0.8\textwidth]{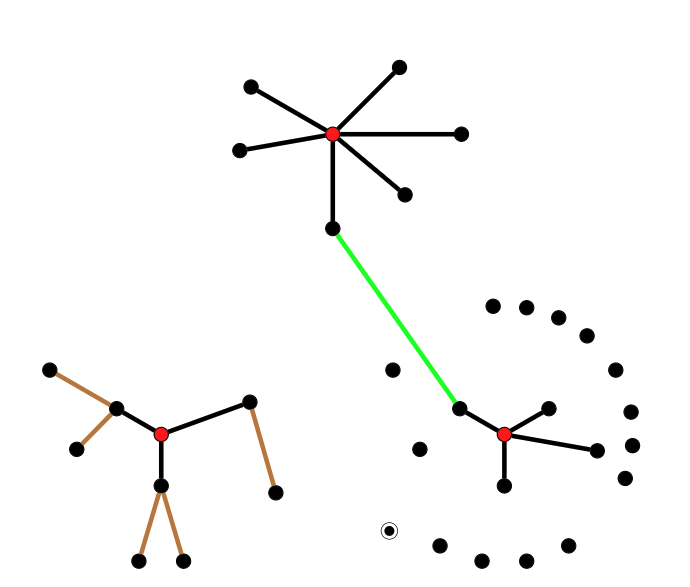}
  \end{subfigure}
  \caption{This figure depicts the Phase 2 of the Cluster Growing algorithm. The left figure shows the situation before growth and the right figure shows the situation after growth. There are three clusters denoted by $C_1$ (top), $C_2$ (left), $C_3$ (right). A red disk marks the center within each cluster. Suppose that $C_1$ is sampled during Phase 2, whereas the remaining two clusters remain non-sampled. Further suppose that $C_2$ is a low-degree cluster and $C_3$ is a high-degree cluster. So we add edges from $C_2$ to  all its neighbors (one edge per neighbor) to the danner; these are shown as brown thick edges. 
  The cluster $C_3$ is adjacent to $C_1$ and hears from it via the green edge; this green edge is added both to the cluster $C_1$ and to the danner $H$. Furthermore, cluster $C_3$ joins cluster $C_1$. 
  }\label{fig:phase2}
\end{figure}

\begin{lemma}
\label{lem:algo2hdldclusters}
    In Algorithm~\ref{alg:part2hd}, for every low-degree cluster $C \in \calc_1 \setminus \calc_2$, $|\n_2(\cen(C))|= O(n^{2\delta})$, w.h.p. 
\end{lemma}
\begin{proof}
    Suppose $|\n_2(w)| \ge c_1 n^{2\delta} \ln n$ for some constant $c_1$. Then w.h.p., for some constant $c_2$ (that depends on $c_1$) there are at least $c_2 n^{\delta} \ln n$ nodes in $\n_2(w)$ that are centers of clusters in $\calc_1$. 
    This implies that w.h.p.~there is at least one node $v \in \n_2(w)$ that is the center of some cluster in $\calc_2$.
    Such a node $v$ informs all 2-hop neighbors (in Step 3), which means this $\cen(C)$ will receive this message, and cluster $C$ will be classified as a high-degree cluster -- a contradiction.  
\end{proof}

\begin{lemma}\label{lem:part2hd}
    Algorithm~\ref{alg:part2hd} takes $O(1)$ rounds and uses $\ot(n^{1+\delta})$ messages.
\end{lemma}
\begin{proof}
    We first bound the round complexity of the algorithm. Note that communication occurs only in Steps 3, 6, 10, and 12. Steps 6 and 10 simply involve sending messages to neighbors and can be executed in 1 round each. Steps 3 and 12 involve executing the depth-2 BFS tree subroutine, and as shown in Section \ref{subsection:fastSubroutines}, this takes $O(1)$ rounds. The remaining steps involve only local computation.

    We now bound the message complexity as follows.
    In Step 3, for each $C \in \calc_2$, we could send as many as $O(n)$ messages (which is a trivial bound). Since $|\calc_2| = O(n^{1-2\delta})$ w.h.p., this is a total of $O(n^{2-2\delta}) = O(n^{1+\delta})$ messages, the latter bound holds because $\delta \ge 1/3$.
    Step 6 incurs at most $n$ messages because each child of a cluster sends at most one message.
    Similarly, $O(n)$ is an upper bound on the number of messages sent in Step 10. 
    Step 12 requires $\ot(n^{1+\delta})$ messages because $\cen(C)$ has as at most $\ot(n^{2\delta})$ 2-hop neighbors by Lemma~\ref{lem:algo2hdldclusters}, and there are $\Theta(n^{1-\delta})$ such clusters w.h.p. Here we also use the fact that the number of messages used by the depth-2 BFS tree subroutine is linear in the number of nodes in the tree.
\end{proof}

\noindent
\textbf{Low $\delta$ case:}
When $\delta\in[0,\frac{1}{3}]$, it is message-inefficient for the center of a sampled cluster $C$ to inform its 2-hop neighbors that $C$ has been sampled. Specifically, Step 3 in Algorithm \ref{alg:part2hd} uses $O(n^{2-2\delta})$ messages, which is bounded above by $O(n^{1+\delta})$ \textit{only for large} $\delta$, i.e., $\delta \ge 1/3$.
To obtain the $O(n^{1+\delta})$ message complexity even for small $\delta$,
unsampled clusters have to learn if there is a neighboring sampled cluster in a more message-frugal manner. 
This challenge is overcome in Algorithm \ref{alg:part2ld}.
Towards this goal, each unsampled cluster $C \in \calc_1 \setminus \calc_2$ considers all messages received in Algorithm \ref{alg:part1}, from sampled clusters in $\calc_1$.
More specifically, recall that we use $M_1(w)$ to denote the set of \id s of clusters in $\calc_1$ that sent $w$ a message in Algorithm \ref{alg:part1} (see Steps 5 and 7 in Algorithm \ref{alg:part1}). 
Each center of a cluster $C \in \calc_1 \setminus \calc_2$ (i.e., an unsampled cluster) then gathers the \id s of all clusters in $\calc_1$ that sent a message to some node $w \in C$ in Algorithm \ref{alg:part1}.
This is done in Step 5 of Algorithm \ref{alg:part2ld} by each node $w \in C$, $w \not= \cen(C)$, simply sending the \id s in $M_1(w)$ one-by-one to $\cen(C)$. This is still round-efficient because w.h.p.~$|M_1(w)| = O(n^{\delta}) = O(n^{1-2\delta})$, with the latter equality being true for $\delta \le 1/3$. This is also message-efficient, again because $|M_1(w)| = O(n^{\delta})$, and so $O(n^{1+\delta})$ messages are sent in this step.
$\cen(C)$ computes the set $M_1(C)$ of \id s of clusters sampled in Algorithm \ref{alg:part1} that sent the cluster $C$ a message. 
If this set is large, i.e., at least $\frac{1}{2} n^{\delta} \ln n$ in size, then the cluster $C$ is classified as a \textit{high-degree cluster}. Such a cluster can be confident that w.h.p.~at least one of the clusters with \id\ in $M_1(C)$ is sampled in Algorithm \ref{alg:part2ld} and belongs to $\calc_2$ (Lemma \ref{lem:algo2hdcluster}). Cluster $C$ can then find and connect to one such cluster $C'$ (Steps 8-10). On the other hand, if $|M_1(C)|$ is small, then we show that it must be the case that the total number of neighbors of cluster $C$ is small (Lemma \ref{lem:algo3ldclusters}). 
In this case (as in Algorithm \ref{alg:part2hd}), we want to grow cluster $C$, i.e., add edges connecting cluster $C$ to each of its neighbors $w$, to the danner $H$.
For this purpose we use the subroutine $\gc(C)$ defined earlier. This subroutine uses $O(\sqrt{D})$ rounds and $O(|V(C)| + D)$ messages, where $D$ is the size of the neighborhood of $C$.
In Lemma \ref{lem:algo3ldclusters} we show that $|\n(C)| = \ot(n^{2\delta})$, w.h.p.
Combining this with the round and message complexity of $\gc(C)$, we see that each low-degree cluster $C$ can connect to all its neighbors in $\ot(n^{\delta})$ rounds and $\ot(|V(C)| + n^{2\delta})$ messages.
Since $\delta \in [0, 1/3]$, this yields a round complexity of $\ot(n^{1-2\delta})$ rounds.
To get a bound on the overall message complexity, we sum over all clusters $C \in \calc_1$ and get an $\ot(n^{1 + \delta})$ bound on the message complexity using the fact that the $|\calc_1| = O(n^{1-\delta})$ w.h.p.

\setcounter{algorithm}{1}
\renewcommand{\thealgorithm}{\arabic{algorithm}(b)}
\begin{algorithm}
\caption{Cluster Growing: Phase 2 [Low Delta]}\label{alg:part2ld}
\begin{algorithmic}[1]
\Require $\delta\in [0,\frac{1}{3}]$; $G=(V,E)$, $\calc_1$ and $H$ are the set of clusters and partial danner output by Algorithm \ref{alg:part1}.
\Ensure a set $\calc_{2}$ of clusters, partially constructed danner $H$
\State Independently sample each cluster $C \in \calc_1$ with probability $n^{-\delta}$.
\Statex \textbf{Notation:} $\calc_2 \subseteq \calc_1$ denotes the set of sampled clusters; for each cluster $C \in \calc_2$, each node $w \in V(C) \setminus \{\cen(C)\}$ is a child of $\cen(C)$.
\Statex
    \For{$C\in\calc_1$} 
    \State $\cen(C)$ broadcasts information on whether $C$ is sampled to all its children. 
    \EndFor
    \Statex
    \For{$C\in \calc_{1} \setminus \calc_{2}$} \Comment{$C$ is not sampled}
        \State Each node $w \in V(C)\setminus \{\cen(C)\}$ transmits \textit{all} the elements in $M_1(w)$ to $\cen(C)$ along edge $\{w, \cen(C)\}$.
        \State $\cen(C)$ computes $M_1(C)$, a maximal subset of $\cup_{w} M_1(w)$ with unique \id s.
        \If{$|M_1(C)| \geq \frac{1}{2}n^{\delta}\ln n$ }\Comment{$C$ is a \textbf{high-degree cluster}}
            \State $\cen(C)$ chooses $X \subseteq M_1(C)$, consisting of the smallest $\lfloor \frac{1}{2}n^{\delta}\ln n\rfloor$ \id s from $M_1(C)$. 
            \State for each $(\id_{C'}, e) \in X$, $\cen(C)$ sends a message along edge $e$ to check if $C' \in \calc_2$.
            \State On finding $C' \in \calc_2$, $\cen(C)$ connects to cluster $C'$ along edge $e$ and adds edge $e$ to $H$.
        \Else \Comment{$C$ is a \textbf{low-degree cluster}}
            \State \colorbox{gray!20}{$\cen(C)$ calls the subroutine $\gc(C)$}  \Comment{Refer to~\ref{subsection:fastSubroutines}}
            \State Edges returned by this subroutine are added to $H$.
        \EndIf
    \EndFor
\end{algorithmic}
\end{algorithm}

\begin{lemma}
\label{lem:algo3ldclusters}
    In Algorithm \ref{alg:part2ld}, if $C$ is a low-degree cluster then $|\n(C)|= O(n^{2\delta} \ln n)$, w.h.p. 
\end{lemma}
\begin{proof}
Consider an arbitrary low-degree cluster $C$.
Suppose $|\n(C)| \ge c n^{2\delta} \ln n$ for a large enough constant $c$. Then, by an application of Chernoff bounds, w.h.p.~at least $\frac{1}{2} n^{\delta} \ln n$ nodes are sampled in Algorithm \ref{alg:part1} and become part of the cluster set $\calc_1$. Let $S \subseteq \n(C)$ denote this subset of neighbors of $C$ who have been sampled in Algorithm \ref{alg:part1}. Consider a node $w \in S$ and let $w' \in V(C)$ be a neighbor of $w$ in $C$. 

If $w$ does not send $w'$ a message $\m_1(\id_{w'})$ in Step 5 of Algorithm \ref{alg:part1}, it must mean that $w'$ has at least $\lceil n^{2\delta} \rceil$ neighbors. In that case, w.h.p.~$w'$ will receive at least $\frac{1}{2} n^{\delta} \ln n$ messages from its neighbors in Step 5 of Algorithm \ref{alg:part1}. 
This contradicts the fact that $C$ is a low-degree cluster.
Hence, every node $w \in S$ must send its neighbor in $C$ a message in Step 5 of Algorithm \ref{alg:part1}. Since $|S| \ge \frac{1}{2} n^{\delta} \ln n$, it must mean that $|M_1(C)| \ge \frac{1}{2} n^{\delta} \ln n$, again contradicting the fact that $C$ is a low-degree cluster.
This implies that $|\n(C)| < c n^{2\delta} \ln n$ and the lemma follows.  
\end{proof}

\begin{lemma}
\label{lem:algo2hdcluster}
    At the end of Algorithm~\ref{alg:part2hd}, each cluster $C \in \calc_1 \setminus \calc_2$ (i.e., a non-sampled cluster) that is designated a high-degree cluster will connect to a cluster $C' \in \calc_2$ (i.e., a sampled cluster), w.h.p.
\end{lemma}
\begin{proof}
Consider a cluster $C \in \calc_1 \setminus \calc_2$ that is designated a high-degree cluster.
In Steps 8-9 in Algorithm \ref{alg:part2ld}, $C$ contacts a set $X$ of neighboring nodes that were sampled in Algorithm \ref{alg:part1}, where $|X| = \lfloor \frac{1}{2}n^{\delta}\ln n \rfloor$. 
Note that each node in $X$ is the center of a cluster in $\calc_1$.
By applying Chernoff bounds we get that w.h.p.~at least one of the clusters in $\calc_1$ whose center belongs to $X$ is sampled in Algorithm \ref{alg:part2ld}.
In Step 10 of Algorithm \ref{alg:part2ld}, cluster $C$ connects to one such cluster $C'$.
\end{proof}

\begin{lemma}\label{lem:part2ld}
    Algorithm~\ref{alg:part2ld} takes $\ot(n^{1-2\delta})$ rounds and uses $\ot(n^{1+\delta})$ messages.
\end{lemma}
\begin{proof}
    We first bound the running time of the algorithm. 
    We note that Steps 6, 7, and 8 involve local computations; we now focus on the remaining steps.
    Step 3 requires $O(1)$ rounds.
    Step 5, however, requires $\ot(n^{\delta})$ rounds (which is $\ot(n^{1-2\delta})$ for $\delta \in [0, \frac{1}{3}]$), because $|M_1(w)| = O(n^\delta)$ for each node $w$ (refer to Lemma~\ref{lem:alg1NS}).
    Step 9 takes $\ot(n^{\delta})$ rounds because in the worst case all nodes in $X$ may have to be contacted via the same child of $\cen(C)$.
    Step 10 takes $O(1)$ rounds because $\cen(C)$ will only contact a single neighboring cluster $C'$ in this step.
    Finally, Steps 12-13 take $\ot(n^{\delta}) = \ot(n^{1-2\delta})$ rounds. This follows from the fact that every low-degree cluster has $\ot(n^{2\delta})$ neighbors (Lemma \ref{lem:algo3ldclusters}) and from the round complexity of the $\gc(C)$ subroutine (Lemma \ref{lemma:GCComplexity}).
    
    We now bound the message complexity of the algorithm.
    Step 3 uses $n$ messages because each node in the graph can be a child in at most one cluster and thus can receive at most one message.
    Step 5 incurs $\ot(n^{1+\delta})$ messages because each node $w$ sends $\ot(n^{\delta})$ messages (see Lemma~\ref{lem:alg1NS}). 
    Step 9 requires $\ot(n)$ messages, taking into account $|X| = \ot(n^{\delta})$ and the existence of up to $O(n^{1-\delta})$ clusters in $\calc_1$.
    Step 10 requires $O(n^{1-\delta})$ messages because each of at most $O(n^{1-\delta})$ clusters send $O(1)$ messages.
    Step 12-13 costs $\ot(n^{1+\delta})$ messages because each cluster contributes $\ot(n^{2\delta})$ edges (refer to Lemma~\ref{lem:algo3ldclusters}), and there are $O(n^{1-\delta})$ such clusters. This analysis step also depends on the linear message complexity of the $\gc(C)$ subroutine (Lemma \ref{lemma:GCComplexity}).
\end{proof}

\begin{lemma}
\label{lem:algo2Correctness}
    After Algorithm Cluster Growing: Phase 2 (refer to Algorithm~\ref{alg:part2hd}, \ref{alg:part2ld}) completes (a) there are $\Theta(n^{1-2\delta})$ clusters w.h.p., (b) every cluster is a tree with $O(1)$ diameter and all cluster edges belong to $H$, and (c) for every cluster $C \in \calc_1 \setminus \calc_2$ that is designated low-degree, there is one edge in $H$ connecting cluster $C$ to each of its neighbors, and (d) $H$ contains $\ot(n^{1+\delta})$ edges w.h.p. 
\end{lemma}
\begin{proof}
    (a) Lemma~\ref{lem:algo1Correctness} shows that after completing Algorithm \ref{alg:part1} there are $O(n^{1-\delta})$ sampled clusters w.h.p. Since clusters are further sampled with probability $n^{-\delta}$ in Phase 2, by applying Chernoff bounds we see that the number of sampled clusters at the end of Phase 2 is $O(n^{1-2\delta})$ w.h.p. 
    (b) As shown in Lemma~\ref{lem:algo1Correctness}, the clusters that are provided as input to Phase 2 are star graphs. In Phase 2, non-sampled clusters merge into neighboring sampled clusters via edges (Step 10 in Algorithm \ref{alg:part2hd} and Step 10 in Algorithm \ref{alg:part2ld}), resulting in clusters with $O(1)$ diameter. All new cluster edges are added to $H$ during these steps.   
    (c) For each low-degree cluster $C$ and each neighbor $w\in\n(C)$, we add an edge from $C$ to $w$ to the danner (see Step 12 in Algorithm \ref{alg:part2hd} and Steps 12-13 in Algorithm \ref{alg:part2ld}). 
    (d) The process of merging a high-degree unsampled cluster with a neighboring sampled cluster
    contributes a total of $O(n^{1-\delta})$ edges to $H$ due to the presence of $O(n^{1-\delta})$ such clusters, each contributing one edge to $H$. (See Step 10 in Algorithm \ref{alg:part2hd} and Step 10 in Algorithm \ref{alg:part2ld}.)
    The process of each low-degree unsampled cluster adding an edge to each neighbor contributes a total of $\ot(n^{1+\delta})$ edges into $H$.
    This is because there are $O(n^{1-\delta})$ clusters, and each cluster adds $\ot(n^{2\delta})$ edges, as specified by Lemmas~\ref{lem:algo2hdldclusters} and  \ref{lem:algo3ldclusters}.
    (See Step 12 in Algorithm \ref{alg:part2hd} and Steps 12-13 in Algorithm \ref{alg:part2ld}.)
\end{proof}

\subsection{Cluster Merging}
\label{subsection:CM}

The cluster merging algorithm in this subsection is similar to the corresponding steps in the Gmyr-Pandurangan \ktone \congest danner algorithm \cite{GmyrP18}, with some key differences in the analysis, which we point out. 

The subsequent paragraphs present a high-level overview of the \textit{danner} construction. Beginning with a graph $G = (V, E)$ where $V$ denotes the set of nodes and $E$ the set of edges, the \textit{danner} construction generates a corresponding \textit{danner} $H$. The algorithm incorporates a parameter $\delta$ that governs the trade-off between the algorithm's time and message complexity, as well as the trade-off between the diameter and the number of edges in $H$. A pivotal step involves categorizing nodes into two groups based on their degrees: low-degree and high-degree nodes. This categorization is essential for optimizing message efficiency in the algorithm.

The initialization involves setting $V(H)$ to $V(G)$, $E(H)$ to an empty set, and defining $V_{high}$ as the set of high-degree nodes (with degree $> n^{\delta}$). The fundamental concept behind the algorithm is as follows: low-degree nodes and their incident edges can be directly incorporated into \textit{danner} $H$. To address high-degree nodes, the algorithm establishes a dominating set that covers these nodes by randomly sampling approximately $n^{(1-\delta)}$ nodes. Each node serves as a center with a probability of $p=c\log n/n^{\delta}$, where $c\geq1$ and $p<1$. The set of centers, denoted as $C$, functions as the dominating set, with a bounded size of $\tilde{O}(n^{1-\delta})$.

Each node $v$ contributes edges to \textit{danner} $H$ connecting it to its min$\{\deg(v), n^\delta\}$ neighbors with the lowest identifiers. High-degree nodes are linked to a center in $H$, and the diameter of each fragment in $H$ is limited to $\tilde{O}(n^{1-\delta})$. The algorithm employs the FindAny procedure from KKT~\cite{king2015construction} to facilitate a distributed Bor\r{u}vka-style merging of fragments in the subgraph $\hat{H}$ induced by high-degree nodes $V_{high}$ and centers $C$. In each merging phase, fragments use FindAny to efficiently identify an outgoing edge, which is then added to \textit{danner} $H$. The entire merging process requires only $O(\log n)$ phases to amalgamate all fragments into a connected graph, with a total of $\tilde{O}(\min\{m, n^{1+\delta}\})$ messages. The algorithm accomplishes the construction of such a \textit{danner} in $\tilde{O}(n^{1-\delta})$ rounds and ~{O}$(\min\{m, n^{1+\delta}\})$ messages.

Recall that $\calc_2$ is the set of clusters returned by Phase 2 of the Cluster Growing algorithm (Algorithms \ref{alg:part2hd} and \ref{alg:part2ld}).
In Algorithm~\ref{alg:part3}, let $V(\calc_2)$ denote the set of vertices belonging to clusters in $\calc_2$, i.e., $V(\calc_2) = \cup_{C \in \calc_2} V(C)$. The steps of Algorithm \ref{alg:part3} are performed on two induced subgraphs, $G[V(\calc_2)]$, which we denote by $\hat{G}$ and $H[V(\calc_2)]$, which we denote by $\hat{H}$. 
The algorithm executes a distributed Bor\r{u}vka-style merging of the connected components of $\hat{H}$ using edges from the underlying graph $\hat{G}$. To do this in a manner that is both round and message efficient, we employ the \fa\ algorithm of KKT~\cite{king2015construction}. During each merging phase, \fa\ is employed by each connected component to efficiently locate an outgoing edge, which is then added to the danner $H$. 
Specific properties of the \fa\ algorithm are described in Theorem \ref{thm:kkt} below.
The process of finding an outgoing edge is coordinated by a leader, elected within each component. For this purpose, we use the leader election algorithm from \cite{KuttenPPRTJACM2015} that is both round and message efficient.
Theorem \ref{thm:leader} below specifies the properties of this leader election algorithm.
The entire process requires only $\log n$ iterations to merge all fragments into a set of maximally connected components, i.e., reach a stage where no further merging is possible.  This takes only $\log n$ iterations because in each iteration, each connected component with an outgoing edge merges with at least one other connected component.
See \cite{king2015construction} for further details of how \fa\ works and how it is used to merge connected components. 
The analysis below shows that the diameter of every connected component before every iteration is bounded above by $\ot(n^{1-2\delta})$ w.h.p. Thus, this is an upper bound on the number of rounds it takes for a leader to coordinate the process of finding an outgoing edge. So waiting for $\ot(n^{1-2\delta})$ rounds in each iteration ensures that all the iterations proceed in lock-step. 

\setcounter{algorithm}{2}
\renewcommand{\thealgorithm}{\arabic{algorithm}}
\begin{algorithm}[H]
\caption{Cluster Merging}\label{alg:part3}
\begin{algorithmic}[1]
\Require $G=(V,E)$, partially constructed danner $H$, set $\calc_2$ of clusters returned by Algorithms \ref{alg:part2hd} or \ref{alg:part2ld}
\Ensure fully constructed danner $H$
\For{$i = 1$ to $\log n$}\Comment{Do the following steps in parallel in each connected component $K$ of $\hat{H}$.}
    \State Elect a leader using the algorithm from Theorem~\ref{thm:leader}. 
    \State \multiline{Using the algorithm \fa\ from Theorem~\ref{thm:kkt} to find an edge in $\hat{G}$ leaving $K$. The leader elected in Step 2 coordinates this process. If such an edge exists, add it to $H$ and $\hat{H}$.}
    \State \multiline{Wait until $\ot(n^{1-2\delta})$ rounds have passed in this iteration before starting the next iteration.\Comment{To synchronize the execution between the connected components.}}  
\EndFor
\end{algorithmic}
\end{algorithm}

\subsubsection{Analysis.}
The Cluster Merging algorithm (Algorithm \ref{alg:part3}) relies on two well-known previously designed algorithms.
The first algorithm, \fa\, is the core of the KKT MST algorithm \cite{king2015construction}. As shown in the theorem below, given a connected component $H$ within a graph $G$, it efficiently identifies an outgoing edge from $H$, if such an edge exists. The natural algorithm for this task would be for each node $v$ in $H$ to scan its neighborhood and identify a neighbor outside $H$. Then, all nodes $v$ in $H$ can upcast one identified edge each to the leader of $H$. Finally, the leader can pick one edge from among these. The problem with this algorithm is that the step that requires $v$ to scan its neighbors is extremely message-inefficient and could require $\Omega(m)$ edges. KKT overcame this issue by cleverly using random hash functions with certain specific properties.

\begin{theorem}[KKT~\cite{king2015construction}]
\label{thm:kkt}
    Consider a connected subgraph $H$ of a graph $G$. An algorithm \fa\ in the \ktone \congest model exists that w.h.p.~outputs an arbitrary edge $\id$ in $G$ leaving $H$ if such an edge exists and $\emptyset$ if no such edge exists. This algorithm takes $\ot(D(H))$ rounds and $\ot(E(H))$ messages.
\end{theorem}

\noindent
We also need an efficient leader election algorithm because we need each connected component $H$ to have a leader that can coordinate the process of finding an outgoing edge. We use the following theorem stated in Gmyr and Pandurangan \cite{GmyrP18}, which in turn is a reformulation of Corollary 4.2 in the paper by Kutten, Pandurangan, Peleg, Robinson, and Trehan \cite{KuttenPPRTJACM2015}.

\begin{theorem}[\cite{KuttenPPRTJACM2015}]
\label{thm:leader}
    There exists an algorithm in the \ktzero \congest model that, for any graph $G$, elects a leader in $O(D(G))$ rounds and utilizes $\ot(E(G))$ messages, w.h.p.
\end{theorem}

\noindent
The Gmyr-Pandurangan analysis requires two key properties to hold \textit{before} the Cluster Merging algorithm: (i) there is a set $\calc$ of clusters, each with constant diameter and (ii) the partially constructed danner $H$ contains all the edges belonging to the clusters along with \textit{all} edges incident on nodes not in clusters. 
If these two properties hold, then they can show that the following crucial property holds \textit{after} the Cluster Merging algorithm:
\begin{quote}
    Before each iteration of the algorithm and after the algorithm ends, the sum of the diameters of all the connected components in $\hat{H}$ is $O(|\calc|)$.
\end{quote}

\noindent
After Phase 2 of our Cluster Growing algorithm ends, we do have a set $\calc_2$ of clusters, with each cluster $C \in \calc_2$ having constant diameter.
However, we do not have the second property.
This is because some nodes not in clusters in $\calc_2$ belong to low-degree clusters not sampled in Phase 2 of the Cluster Growing algorithm. Specifically, consider a cluster $C \in \calc_1 \setminus \calc_2$ such that cluster $C$ is designated as a low-degree cluster in  Phase 2 of the Cluster Growing algorithm. Here, we refer to both versions of our algorithm, i.e., Algorithm \ref{alg:part2hd} and \ref{alg:part2ld}.
Nodes in $C$ may have only a small number of incident edges belonging to $C$ and therefore to $H$. So Property (ii) above, which is required by the analysis of the Gmyr-Pandurangan algorithm, may not hold.
However, we know that for every such cluster $C$, we add edges connecting $C$ to each of its neighbors $w \in \n(C)$ to the danner. This means we can treat each low-degree cluster $C \in \calc_1 \setminus \calc_2$ as a \textit{super node} and contract it. Each super node now has the property that all incident edges are in $H$. Given that after Phase 1 of our Cluster Growing algorithm, we have set aside a set of nodes (see Step 11 in Algorithm \ref{alg:part1}) and added all incident edges to $H$, we now have both properties needed by the Gmyr-Pandurangan analysis. As a result, we obtain the following lemma. 
It is worth highlighting that since $|\calc_2| = \ot(n^{1-2\delta})$, the sum of the diameters of the connected components in $\hat{H}$ is also $\ot(n^{1-2\delta})$.
This is in contrast with the corresponding Gmyr-Pandurangan lemma that obtains a weaker 
$\ot(n^{1-\delta})$ because that is the number of clusters they have before starting the Cluster Merging algorithm.

\begin{lemma}
Let $K_1, \ldots, K_r$ be the connected components of $\hat{H}$ before any iteration of the
loop in Algorithm \ref{alg:part3} or after the final iteration. It holds $\sum_{i=1}^r diam(K_i) = \ot(n^{1-2\delta})$, w.h.p.
\end{lemma}

\noindent
The rest of the analysis is identical to that of Gmyr and Pandurangan~\cite{GmyrP18} and we obtain the following lemmas.
\begin{lemma}\label{lem:part3}
    Algorithm~\ref{alg:part3} computes a danner in $\ot(n^{1-2\delta})$ rounds and using $\ot(min\{m,n^{1+\delta}\})$ messages w.h.p.  
\end{lemma}

\begin{lemma}
\label{lem:algo3correctness}
    Algorithm~\ref{alg:part3} computes a danner with $D(H)\le D(G)+\ot(n^{1-2\delta})$ and with $\ot(min\{m,n^{1+\delta}\})$ edges, w.h.p.
\end{lemma}

\noindent
With Lemmas~\ref{lem:part1}, \ref{lem:part2hd}, \ref{lem:part2ld}, and \ref{lem:part3} in place, it is easy to see that the Danner Algorithm (including Algorithm~\ref{alg:part1}, \ref{alg:part2hd}, \ref{alg:part2ld}, and \ref{alg:part3}) takes $\ot(n^{1-2\delta})$ rounds and sends $\ot(n^{1+\delta})$ messages w.h.p. With Lemmas ~\ref{lem:algo1Correctness}, \ref{lem:algo2Correctness}, and \ref{lem:algo3correctness} in place, after the algorithm terminates it holds that $H$ has $\ot(n^{1+\delta})$ edges and $\ot(D+n^{1-2\delta})$ diameter, where $\delta\in[0,\frac{1}{2}].$ 
As a result, we obtain the following theorem directly. 
\begin{theorem}
\label{thm:danner}
    The Danner Algorithm (including Algorithm~\ref{alg:part1}, \ref{alg:part2hd}, \ref{alg:part2ld}, and \ref{alg:part3}) takes $\ot(n^{1-2\delta})$ rounds and sends $\ot(n^{1+\delta})$ messages w.h.p. After the algorithm terminates, it holds that $H$ has $\ot(min\{m,n^{1+\delta}\})$ edges and $\ot(D+n^{1-2\delta})$ diameter w.h.p., where $\delta\in[0,\frac{1}{2}].$
\end{theorem}

\section{Applications}
In this section, we illustrate how the danner construction, introduced in Section~\ref{section:danner}, can effectively establish trade-off results for many fundamental problems in the field of distributed computing.

\subsection{\bc, Leader Election, and Spanning Tree}
Some immediate implications of Theorem \ref{thm:danner} are that \bc, leader election, and spanning tree construction can be solved fast and using a few messages. To address \bc, a danner can be initially constructed, followed by implementing a straightforward, flooding-type algorithm that leverages all the edges within the danner. In the context of leader election, Theorem~\ref{thm:leader} from Kutten et al.~\cite{KuttenPPRTJACM2015} can be applied to the computed danner, offering a rapid and effective resolution. Lastly, for the task of spanning tree construction, a leader can be elected to execute a distributed breadth-first search on the danner, culminating in the efficient creation of the spanning tree. The ensuing theorem encapsulates these advancements.

\begin{theorem}
\label{thm:broadcast}
Given any connected graph $G$ and any $\delta\in[0,\frac{1}{2}]$, there are algorithms for solving \bc, leader election, and spanning tree construction in the \kttwo \congest{} model in $O(D+n^{1-2\delta})$ rounds using $\tilde{O}(min\{m, n^{1+\delta}\})$ messages w.h.p.
\end{theorem}

Since the danner and \bc problems are such important algorithmic primitives in distributed computing, Theorems \ref{thm:danner} and Theorem \ref{thm:broadcast} have important implications for other problems.
Using these results, we design fast, low-message algorithms for MST and $(\Delta+1)$-coloring in the \kttwo model.
MST has a well-known round complexity lower bound: it cannot be solved in fewer than $\tilde{\Omega}(D + \sqrt{n})$ rounds in the \congest model \cite{PelegRubinovichSICOMP2000,sarma2012distributed}.
The lower bound argument in \cite{sarma2012distributed} uses a novel reduction from the 2-party communication complexity problem \textsc{SetDisjointness}. 
We observe that the proof in \cite{sarma2012distributed} works even in the \ktrho \congest model, for any $\rho$, $0 \le \rho \le 2\log \frac{\sqrt{2n}}{4} + 2$, where $n$ is the size of the network.
Thus, we cannot beat this round complexity lower bound even as we increase the radius $\rho$ of initial knowledge. So our goal is to design an algorithm that matches this round complexity lower bound while substantially reducing the number of messages as $\rho$ increases. 

\subsection{Fast low-message MST}
The danner result in \ktone \congest model of Gmyr and Pandurangan~\cite{GmyrP18} is applied to build an efficient MST algorithm. Roughly speaking, given a connected graph $G$ and any $\delta\in[0,\frac{1}{2}]$, the MST of $G$ construction includes three steps. 
In the first step, a spanning tree of $G$ of depth $\ot(D+n^{1-\delta})$ is built on the danner, where the danner has diameter $\ot(D+n^{1-\delta})$. 
When $m\le n^{1+\delta}$, the MST is computed by calling the singularly optimal MST of Pandurangan et al.~\cite{pandurangan2017time} on $G$, where $m$ is the number of edges in $G$. This algorithm takes $\ot(D+\sqrt{n})$ rounds and requires $\ot(m)$ messages. 

Otherwise, the MST is computed with the following two steps. 

In the second step, the MST algorithm executes a Controlled-GHS procedure as described in~\cite{pandurangan2017time}. This procedure requires $\ot(m)$ messages and $\ot(n^{1-\delta})$ rounds to return at most $n^{\delta}$ MST-fragments with each has diameter $O(n^{1-\delta})$. Using KKT~\cite{king2015construction} can reduce the message complexity to $\ot(n)$ without worsening the running time. 

Step 3 merges the remaining $n^\delta$ MST fragments efficiently by the same procedure using $\log n$ iterations. 

It is clear that if we substitute the danner part in the first step, with our \dn\ result~\ref{thm:danner}, then keep steps 2 and 3 unchanged, we will get the following result.

\begin{theorem}\label{thm:MST}
    Given a connected graph $G$ and any $\delta\in[0,\frac{1}{4}]$, an MST of $G$ and be computed in $\ot(D+n^{1-2\delta})$ rounds, while using $\ot(n^{1+\delta})$ messages, w.h.p. 
\end{theorem}
For $\delta=\frac{1}{4}$, we get the following corollary with optimal running time with fewer messages, compared to the danner algorithm in \ktone of Gmyr and Pandurangan~\cite{GmyrP18}.

\begin{corollary}
    There is an algorithm that can compute an MST in $\ot(D + \sqrt{n})$ rounds using $\ot(n^{1+\frac{1}{4}})$ messages in the \kttwo \congest model.
\end{corollary}

\subsection{Fast low-message {$(\Delta+1)$}-coloring}
\label{section:coloring}

In \cite{PaiPPRPODC2021}, the authors present a $(\Delta+1)$-coloring algorithm in the \ktone \congest model that 
uses $\ot(n^{1.5})$ messages, while running in $\ot(D + \sqrt{n})$ rounds, where $D$ is the graph diameter.
In this section, we present a $(\Delta+1)$-coloring algorithm in the \kttwo \congest model using our danner and \bc results. The specific result we prove is presented in Theorem~\ref{thm:coloring}.

We start with an overview of the randomized non-comparison-based $(\Delta+1)$-coloring algorithm from Pai et al. \cite{PaiPPRPODC2021}. This algorithm applied a simple graph partitioning technique that appeared in the $(\Delta+1)$-coloring algorithm in \cite{ChangFGUZPODC2019}. The Change et al. \cite{ChangFGUZPODC2019} graph partitioning algorithm is as follows. Consider a subgraph $G=(V,E)$ with maximum degree $\Delta$, where $n\ge |V|$. Each vertex $v\in V$ has a palette $\Psi(v)$ and let $k=\sqrt{\Delta}$. 
\begin{description}
\item[Partiton vertex set:] The partition $V=B_1\bigcup \cdots \bigcup B_k \bigcup L$ is defined as follows. Include each vertex $v\in V$ to the set set $L$ with probability $q=\Theta\left(\sqrt{\frac{\log n}{\Delta^{1/4}}}\right)$. Each remaining vertex joins one of $B_1,\hdots,B_k$ uniformly at random. Note that $P[v\in B_i]=p(1-q)$, where $p=\frac{1}{k}=\frac{1}{\sqrt{\Delta}}$.
\item[Partition palette:] Let $C = \bigcup_{v\in V} \Psi(v)$ denote the set of all colors. The partition $C = C_1\bigcup\cdots\bigcup C_k$ is defined by having each color $c\in C$ joins one of the $k$ sets uniformly at random. Note that $P[c\in C_i]=p=\frac{1}{k}=\frac{1}{\sqrt{\Delta}}$.  
\end{description}
Change et al. \cite{ChangFGUZPODC2019} then show that the output of the partitioning algorithm satisfies the following properties, w.h.p., assuming that $\Delta=\omega(\log^2 n)$.
\begin{description}
 \item[(i) Size of Each Part:]$|E(G[B_i])|=O(|V|)$, for each $i\in[k]$. Also, $|L|=O(q|V|)=O(\frac{\sqrt{\log n}}{\Delta^{1/4}})\cdot |V|.$
 \item[(ii) Available Colors in $B_i$: ]For each $i\in[k]$ and $v\in B_i$, let the number of available colors in $v$ in the subgraph $B_i$ is $g_i(v):= |\Psi(v)\cap C_i|$. Then $g_i(v)\ge \Delta_i + 1$, where $\Delta_i :=\max_{v\in B_i}\deg_{B_i}(v)$.
 \item[(iii) Available Colors in L:] For each $v\in L$, define $g_L(v):=|\Psi(v)|-(\deg_G(v)-\deg_L(v))$. Then $g_L(v)\ge \max\{\deg_L(v),\Delta_L-\Delta_L^{3/4}\}+1$ for each $v\in L$, where $\Delta_L:=\max_{v\in V}\deg_L(v)$. Note that $g_L(v)$ represents a lower bound on the number of available colors in the palette of $v$ after all of $B_1,\hdots,B_k$ have been colored.
 \item[(iv) Remaining Degrees:] The maximum degrees of $B_i$ and $L$ are $\deg_{B_i}(v)\leq \Delta_i=O(\sqrt{\Delta})$ and $\deg_L(v)\leq \Delta_L=O(q\Delta)=O(\frac{\sqrt{\log n}}{\Delta^{1/4}})\cdot \Delta$. For each vertex, we have $\deg_{B_i}(v)\leq \max\{O(\log n),O(1/\sqrt{\Delta})\cdot \deg(v)\}$ and $\deg_L(v)\leq \max\{O(\log n), O(q)\cdot \deg(v)\}$.
\end{description}

By utilizing the graph partitioning technique, a danner structure from \cite{GmyrP18}, and a randomized list coloring algorithm by Johansson \cite{johansson1999simple}, the Pai et al. \cite{PaiPPRPODC2021} $\ktone$ $(\Delta+1)$-coloring algorithm takes an $n$-vertex graph $G$ as input, with maximum degree $\Delta$ and diameter $D$ and produces a $(\Delta+1)$-list-coloring of $G$ with the following steps.
\begin{enumerate}
    \item Let $\delta=1/2$, construct a danner $H$, elect a leader broadcasting $O(\log^2 n)$ random bits. 
    \item Each node samples 3 $O(\log n)$-wise independent hash functions based on the $O(\log^2 n)$ random bits: (i) $h_L$ decides join $L$ or not, (ii) $h_{B_i}$ decides joins which $B_i$, and (iii) $h_c$ decides which color joins which $C_i$.
    \item In each $B_i$ in parallel, nodes run a randomized list coloring algorithm by Johansson. 
    \item The induced graph $G[L]$ can be checked if it includes $\ot(n)$ edges by the danner $H$. 
    \item If $G[L]$ includes $\ot(n)$ edges,  the list coloring algorithm by Johansson is executed on $G[L]$.
    \item Otherwise, we recursively run this algorithm on $G[L]$ with the same parameter $n$.
\end{enumerate}

We make the following changes to generalize this result to the \kttwo \congest{} model.

\begin{enumerate}
    \item Let $k=\Delta^{2/3}, q=\frac{\sqrt{\log n}}{\Delta^{1/6}}, \Delta=\omega(\log^3 n),$ in the Change et al. \cite{ChangFGUZPODC2019} graph partitioning algorithm.
    \item Replaces Step (1) from the above algorithm by our danner construction and \bc result.
    \item Using our danner to check Step (4). 
\end{enumerate}

With $k=\Delta^{2/3}, q=\frac{\sqrt{\log n}}{\Delta^{1/6}}$, we update properties (i) and (iv) as follows:
\begin{description}
    \item[updated(i) Size of Each Part:] $|E(G[B_i])|=O(|V|)$, for each $i\in[k]$. Also, $|L|=O(q|V|)=O(\frac{\sqrt{\log n}}{\Delta^{1/6}})\cdot |V|.$
    \item[updated(iv) Remaining Degrees:] The maximum degrees of $B_i$ and $L$ are $\deg_{B_i}(v)\leq \Delta_i=O(\Delta^{1/3})$ and $\deg_L(v)\leq \Delta_L=O(q\Delta)=O(\frac{\sqrt{\log n}}{\Delta^{1/6}})\cdot \Delta$. For each vertex, we have $\deg_{B_i}(v)\leq \max\{O(\log n),O(1/\Delta^{2/3})\cdot \deg(v)\}$ and $\deg_L(v)\leq \max\{O(\log n), O(q)\cdot \deg(v)\}$.
\end{description}
 \begin{lemma}
\label{lem:coloring}
    Updating $k=\Delta^{2/3}, q=\frac{\sqrt{\log n}}{\Delta^{1/6}}, \Delta=\omega(\log^3 n), p = \frac{1}{k}=\frac{1}{\Delta^{2/3}}$, the same proof of Lemma 3.1 in \cite{ChangFGUZPODC2019} goes through to show that properties updated(i), (ii), (iii), and updated(iv) hold w.h.p. 
\end{lemma}

The following lemma is proved in \cite{PaiPPRPODC2021} and given Lemma \ref{lem:coloring}, it goes through without any changes.
\begin{lemma}
\label{lem: limited independence} 
    Properties updated(i), (ii), (iii), and updated(iv) hold w.h.p. when the partition is done with $O(\log n)$-wise independence. 
\end{lemma}
The following lemma is proved in \cite{ChangFGUZPODC2019} and given Lemma \ref{lem: limited independence}. It goes through if we update $\Delta=\log^{3+\alpha} n, \alpha=3+\beta$.
\begin{lemma}
    The algorithm makes $O(1)$ recursive calls w.h.p.
\end{lemma}
\begin{proof}
    Let $\Delta=\log^{3+\alpha} n, \alpha=3+\beta$, as the proof in \cite{ChangFGUZPODC2019}, we can show that \\$\Delta_i=O\left((\log n)^{3+\alpha(5/6)^{i-1}}\right)$, and $|V_i|=O(n/\Delta)\cdot {\Delta_i}=n\cdot O\left((\log n)^{\alpha((5/6)^{i-1}-1)}\right)$. Thus, given that $\alpha=\Omega(1)$ and $i=O(1)$, the condition of $\Delta_i=\omega(\log^3 n)$ for applying Lemma \ref{lem:coloring} must be met. In addition, the condition for $\Delta_i|V_i|=O(n)$ can be rewritten as
    $$3-\alpha+2\alpha\cdot (5/6)^{i-1}\le 0$$
    and
    $$i\ge 1+ \log_{6/5}{\frac{2(3+\beta)}{\beta}}$$
    Given $\alpha=\Omega(1)$, same as $\beta$, we have $1+ \log_{6/5}{\frac{2(3+\beta)}{\beta}}=O(1)$.
\end{proof}

\begin{theorem}
\label{thm:coloring}
There is an algorithm in the \kttwo \congest model that computes a $(\Delta+1)$-coloring using
$\ot(n^{1+\delta})$ messages in $\ot(D+n^{1-2\delta})$ rounds.
\end{theorem}

\begin{proof}
    In Step 1, we build a danner using $\ot(min\{m,n^{1+\delta}\})$ messages in $\ot(D + n^{1-2\delta})$ rounds. In addition, broadcasting $O(\log^2 n)$ random bits takes $\ot(D + n^{1-2\delta})$ rounds using$\ot(min\{m,n^{1+\delta}\})$ messsages with our fast, low-message \bc algorithm, see Theorem \ref{thm:broadcast}. Step 2 is local computation. 

    In Step 3, when we set $k=\Delta^{2/3}$, each vertex $v$ joins $B_i$ with probability $\frac{1}{\Delta^{2/3}}$, for each $i\in [k]$. Consider an edge $e= (u,v)$, we have $P[e\in G[B_i]]=\frac{1}{\Delta^{2/3}}\cdot \frac{1}{\Delta^{2/3}}=\frac{1}{\Delta^{4/3}}$. Let $G[B_i]^e$ denote the number of edges in $G[B_i]$, Thus, $\e[G[B_i]^e]\leq (n\Delta)\cdot \frac{1}{\Delta^{4/3}}=\frac{n}{\Delta^{1/3}}$ . Hence, by linearity of expectation, $\e[\sum_{i\in [k]} G[B_i]^e]=\sum_{i\in [k]}\e[ G[B_i]^e] \leq \frac{n}{\Delta^{1/3}}\cdot \Delta^{2/3}=n\Delta^{1/3}$. We run Johansson's randomized algorithm on each $G[B_i]$ in $O(\log n)$ rounds and takes $\ot(n\Delta^{1/3})$ messages. 

    Step 4 takes $\ot(min\{m,n^{1+\delta}\})$ messages in $\ot(D + n^{1-2\delta})$ rounds w.h.p. by Theorem~\ref{thm:danner}. 

    The above arguments guarantee that Steps 5 and 6 take $\ot(D + n^{1-2\delta})$ rounds and $\ot(min\{m,n^{1+\delta}\})$ messages, w.h.p. The theorem follows.
\end{proof}

This result demonstrates that our danner result has important implications, not just for global problems but for classical local problems as well.

\section{Conclusion}
\label{section:conclusion}

Our main contribution is showing that the round-message tradeoff shown by Gmyr and Pandurangan in the \ktone \congest model can be substantially improved in the \kttwo \congest model. Specifically, we show that there is a danner algorithm in the \kttwo \congest model that runs in $\ot(n^{1-2\delta})$ rounds, using $\ot(\min\{m, n^{1+\delta}\})$ messages w.h.p.
The \dn\ constructed by this algorithm has diameter $\ot(D+n^{1-2\delta})$ and $\ot(\min\{m, n^{1+\delta}\})$ edges w.h.p.
Similar to Gmyr and Pandurangan, we obtain implications of this danner construction for a variety of global problems, namely \bc, LE, ST, and MST, as well as the $(\Delta+1)$-coloring problem (which is a local problem).

Since we don't show lower bounds, it is not clear if the round-message tradeoff we show is optimal. This is open in the \ktone \congest model as well because we don't know if the tradeoff shown by Gmyr and Pandurangan is optimal.
One possible way to improve the tradeoff we show is to construct a constant-spanner, rather than a danner, which imposes a large additive factor $n^{1-2\delta}$ on the diameter of the subgraph. However, it is not clear if a constant-spanner can be constructed in the \kttwo \congest model in $\ot(n^{1-2\delta})$ rounds, using $\ot(\min\{m, n^{1+\delta}\})$ messages. This problem would be a natural follow-up to our current work.

\bibliographystyle{plain}
\thispagestyle{empty}
\bibliography{refs}

\newpage
\appendix
\section{Appendix}
\subsection{A bad example for the AGPV sparsification}
As mentioned earlier, AGPV showed an upper bound of $O(\min\{m, n^{1+c/\rho}\})$ on the message complexity of \bc on an $n$-vertex, $m$-edge graph in the \ktrho \congest model. But, this algorithm can take $\Omega(n)$ rounds in the worst case. The AGPV algorithm starts by performing a deterministic sparsification step that takes 0 rounds (i.e., only local computation is needed by this algorithm), and it reduces the number of edges in the graph to $O(\min\{m, n^{1+c/\rho}\})$.
More precisely, given an $n$-vertex, $m$-edge connected graph $G=(V, E)$ with distinct edges weights, the algorithm marks the heaviest edge in every cycle with length $2\rho$ or less for deletion and outputs a subgraph $\Bar{G} = (V,\Bar{E})$, where $\Bar{E}$ includes all the unmarked edges. Note that in the \ktrho \congest{} model, every node that belongs to a cycle $C$ of length $2\rho$ or less knows $C$, as part of its initial knowledge. So this algorithm requires no communication.
After the sparsification, the algorithm can use any \bc algorithm on $\Bar{G}$ that uses messages proportional to the number of edges in $\Bar{G}$.

In the following, for any positive integer $\rho$, we show a simple example of an unweighted $n$-vertex graph $G$ with $\Theta((n^2)$ edges and diameter $O(1)$. The vertices of the graph are assigned unique \texttt{ID}s and we assume that the weight of each edge $e = \{u, v\}$ is the tuple $(\texttt{ID}_u, \texttt{ID}_v)$, where $\texttt{ID}_u < \texttt{ID}_v$.
When the AGPV sparsification is applied to this graph, we get a graph  
$\Bar{G}$ with $\Omega(n^{1+\frac{1}{2\rho + 1}})$ edges and diameter $\Omega(n)$.  

Let $\rho$ be a positive integer and let $n$ be a multiple of 4. Construct an $n$-vertex undirected graph $G = (V, E)$ as follows. See Figure \ref{fig:badexample}.
\begin{enumerate}
    \item Partition the vertex set $V$ into four parts, $A_i$, $i = 1, 2, 3, 4$, where $$A_i = \left\{\frac{n}{4}\cdot (i-1) + 1, \frac{n}{4}\cdot (i-1) + 2, \ldots,i\cdot\frac{n}{4}\right\}$$ 
    For each vertex $i \in V$, we use $i$ as the \texttt{ID} of the vertex for the AGPV algorithm.   
    \item Add a complete bipartite graph $K_{\frac{n}{4},\frac{n}{4}}$ between $A_1$ and $A_3$.
    \item Add the path $(1, 2, 3, \ldots, \frac{n}{4})$ on the vertices in $A_1$.
    \item Add a perfect matching between vertices in $A_3$ and $A_4$ and between vertices in $A_4$ and $A_2$, as shown in Figure \ref{fig:badexample}.
    \item Add $\frac{1}{4}\cdot (\frac{n}{4})^{1+\frac{1}{2\rho+1}}$ edges between vertices in $A_2$ such that the graph $G[A_2]$ induced by $A_2$ has girth at least $2\rho+1$. Such an edge set exists via a simple probabilistic method proof (see Theorem 6.6 in \cite{mitzenmacher2017probability}).
\end{enumerate}
It is clear that $G$ has $\Theta(n^2)$ edges and diameter 6. 

Now suppose that we apply the AGPV sparsification algorithm \cite{awerbuch1990trade} to $G$, where the weight of each edge $\{u, v\}$ is the ordered pair of \texttt{ID}s of $u$ and $v$ with the lower ID appearing first. Let $\Bar{G}=(V,\Bar{E})$ be the resulting graph.
\begin{lemma}
The graph $\Bar{G}$ has $\Theta(n^{1+\frac{1}{2\rho+1}})$ edges and diameter $\Omega(n)$.
\end{lemma}
\begin{proof}
Every edge $\{i, j\}$, $i \in A_1 \setminus \{1\}$, $j \in A_3$, belongs to the 3-cycle $(i-1, i, j)$ and $\{i, j\}$ is the heaviest edge in this 3-cycle. Therefore, every such edge $\{i, j\}$ is marked for deletion. In addition, every cycle $C$ with length $2\rho$ or less, including edges from $A_2$, must include edges between $A_3$ and $A_4$. For such a cycle, we will mark the edge between $A_3$ and $A_4$ with the heaviest weight for deletion. For such a cycle, this edge always exists because the IDs of the two endpoints of this edge are greater than any ID from $A_2$. In this way, we will not mark any edges in $A_2$ for deletion. $\Bar{G}$ has $\Theta(n^{1+\frac{1}{2\rho+1}})$ edges because $G[A_2]$ includes these many edges as (5). It is clear that $\Bar{G}$ has diameter $\Omega(n)$ because of the shortest path between vertex $\frac{n}{4}$ and vertex 1.     
\end{proof}

\begin{figure}[t]
    \centering
    \includegraphics[width=1\linewidth]{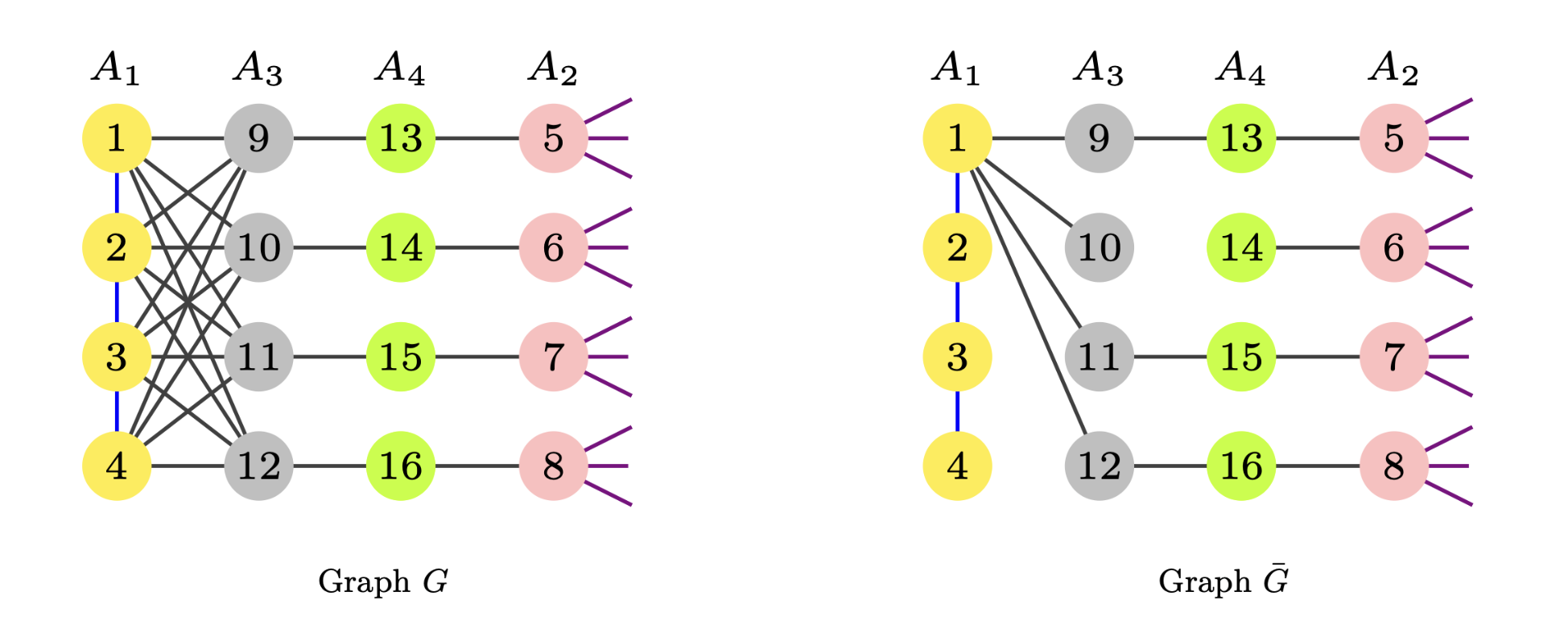}
    \caption{Given $n=16$, construct a graph $G$ with the 5 steps, where $A_1=\{1,2,3,4\},A_2=\{5,6,7,8\}, A_3=\{9,10,11,12\}\text{ and } A_4=\{13,14,15,16\}$. Given such a $G$, we have $\Bar{G}$ after applying AGPV sparsification algorithm \cite{awerbuch1990trade}.}
    \label{fig:badexample}
\end{figure}

\end{document}